\documentclass[journal]{IEEEtran}
\usepackage{amssymb,amsmath,amsthm}
\usepackage{siunitx}
\usepackage{cite}
\usepackage{float}
\usepackage{dblfloatfix}
\usepackage{graphicx} 
\usepackage{epstopdf}         
\usepackage{caption,subcaption}
\usepackage[latin1]{inputenc}
\usepackage{color,soul}

\newtheorem{lemma}{Lemma}
\newtheorem{theorem}{Theorem}
\newtheorem{corollary}{Corollary}

\usepackage{stackengine}
\newcommand\stacksign[2]{%
	\mathrel{\stackunder[2pt]{\stackon[4pt]{$\gtrless$}{$\scriptscriptstyle#1$}}{%
			$\scriptscriptstyle#2$}}}

\begin{document}

\title{Ultra-Reliable Cooperative Short-Packet Communications with Wireless Energy Transfer}
\author{
	\IEEEauthorblockN{	Onel L. Alcaraz López,
						Evelio Martín García Fernández, 
						Richard Demo Souza, and 
						Hirley Alves
					}
	\thanks{O. L. A. López and H. Alves are with Centre for Wireless Communications, University of Oulu, Finland. \{onel.alcarazlopez, hirley.alves\}@oulu.fi.}
	\thanks{E. M. G. Fernández is with Federal University of Paraná (UFPR), Curitiba, Brazil. evelio@ufpr.br.}
	\thanks{R. D. Souza is with Federal University of Santa Catarina (UFSC), Florianopolis, Brazil. richard.demo@ufsc.br.}
	\thanks{This work is partially supported by Academy of Finland (Grant n.303532 and n.307492) and Capes/CNPq (Brazil).}
}
\maketitle

\begin{abstract}
We analyze a cooperative wireless communication system with finite block length and finite battery energy, under quasi-static Rayleigh fading. Source and relay nodes are powered by a wireless energy transfer (WET) process, while using the harvested energy to feed their circuits, send pilot signals to estimate channels at receivers, and for wireless information transmission (WIT). Other power consumption sources beyond data transmission power are considered. The error probability is investigated under perfect/imperfect channel state information (CSI), while reaching accurate closed-form approximations in ideal direct communication system setups. We consider ultra-reliable communication (URC) scenarios under discussion for the next fifth-generation (5G) of wireless systems.
The numerical results show the existence of an optimum pilot transmit power for channel estimation, which increases with the harvested energy.
We also show the importance of cooperation, even taking into account the multiplexing loss, in order to meet the error and latency constraints of the URC systems.
\end{abstract}

\section{Introduction}\label{int}
\subsection{Motivation}
The 5G concept goes beyond broadband connectivity by bringing new kinds of services to life, enabling mission-critical control through ultra-reliable, low-latency links and connecting a massive number of smart devices, enabling the Internet of Things (IoT) \cite{Campbell2017}. It will not only interconnect people, but also interconnect and control a massive number of machines, objects, and devices, with applications having stringent requirements on latency and reliability. In fact, a novel operation mode under discussion for 5G is Ultra-Reliable Communication over a Short Term (URC-S)\cite{Popovski.2014,Johansson.2015,Yilmaz.2015}, which focuses on scenarios with stringent latency, e.g., $\le 10$ms and error probability, e.g., $10^{-3}$, requirements \cite{Popovski.2014}. 

Powering and keeping uninterrupted operation of such massive number of nodes is a major challenge \cite{Zanella.2014}, which could be addressed through the WET technique. WET constitutes an attractive solution because of the coverage advantages of radio-frequency (RF) signals, specially for IoT scenarios where replacing or recharging batteries require high cost and/or can be inconvenient or hazardous (e.g., in toxic environments), or highly undesirable (e.g., for sensors embedded in building structures or inside the human body) \cite{Zhang.2013,Kilinc.2015,Ejaz.2016}. Also, RF signals can carry both energy and information, which enables energy constrained nodes to harvest energy and receive information \cite{Varshney.2008,Grover.2010}, allowing to prolong their lifetime almost indefinitely.

The most important characteristics of WET systems are  \cite{Makki.2016}: i) power consumption of the nodes on the order of $\SI{}{\micro\watt}$; ii) strict requirements on the reliability of the energy supply and of the data transfer; iii) information is conveyed in short packets due to intrinsically small data payloads, low-latency requirements, and/or lack of energy resources to support longer transmissions \cite{Khan.2016}. This agrees well with several URC-S scenarios with stringent latency requirements. Even though, most of the work in the field is made under the ideal assumption of communicating with large enough blocks in order to invoke Shannon theoretic arguments to address error performance, which contradicts the third characteristic of the WET systems. In fact, although performance metrics like Shannon capacity, and its extension to non-ergodic channels, have been proven useful to design current wireless systems since the delay constraints are typically above 10~ms \cite{Devassy.2014}, they are not necessarily appropriate in a short-packet scenario \cite{Durisi.2015}. 
Also, for short-packet scenarios where the sizes of the preamble, the metadata and the data of the frame structure, are of the same order of magnitude, the Shannon capacity becomes an inaccurate metric for assessing the necessary blocklength to achieve a certain reliability. Instead, an essential quantity is the maximum coding rate for which \cite{Polyanskiy.2010} (and references therein) developed nonasymptotic bounds and approximations.

Our goal in this paper is to investigate wireless-powered communication with short packets, which is of great interest for URC-S use cases in 5G systems.  We analyze a cooperative setup, where source and relay nodes are powered by a WET process from the destination, and practical issues such as the finite battery energy, other power consumption sources beyond data transmission, and imperfect CSI, are taken into account.
\begin{table*}[!t]
	\centering	
	\label{symbols}
	\caption{Summary of Main Symbols}
	\begin{tabular}{|ll|ll|}
		\hline        		
		\textbf{Symbol}& \textbf{Definition}&\textbf{Symbol}& \textbf{Definition}	\\ \hline
		$S,R,D$ & Source, Relay and Destination nodes, respectively				& $B_{_{\mathrm{max}}}$ & Battery capacity of $S$ and $R$	\\
		$v$ & WET blocklength	&
		$u$ & Pilot signal blocklength	\\
		$n$& Information blocklength in the DC scheme &
		$n^*$& Optimum information blocklength in the DC scheme \\
		$n_1,n_2$& WIT blocklength in the 1st, 2nd phase with relaying &
		$\nu$& Number of data and pilot symbols in the relaying scheme\\
		$k$ & Message length in bits &
		$h_{ij}$ & Normalized channel in $i\rightarrow j$, $i,j\in\{S,R,D\}$\\
		$g_{ij}$ & Power gain coefficients in $i\rightarrow j$, $i,j\in\{S,R,D\}$&
		$P_{\mathrm{csi}}$ & Pilot signal transmit power for both $S$ and $R$\\
		$T_c$ & Duration of each channel use&
		$E_{i}$ & Energy harvested at $i\in\{S,R\}$\\
		$B_{i}$ & Battery charge in $i\in\{S,R\}$  before transmitting&
		$\eta$ & Energy conversion efficiency\\
		$d_{ij}$ & Distance of the link $i\rightarrow j$, $i,j\in\{S,R,D\}$&
		$\alpha$ & Path loss exponent\\
		$\kappa_{ij}$ & Path loss factor in $i\rightarrow j$, $i,j\in\{S,R,D\}$&
		$P_{i}$ & Transmit power of $i\in\{S,R,D\}$\\
		$\lambda_{i}$ & Threshold for the battery saturation in $i\in\{S,R\}$&
		$P_c$ & Circuit and baseband processing power consumption\\
		$E^t_{i}$ & Available energy for transmission at $i\in\{S,R\}$&
		$x_{i}$ & Signal transmitted by $i\in\{S,R\}$\\
		$\omega_{j}$ & Gaussian noise vector at $j\in\{D,R\}$&
		$\sigma^2_{j}$ & Variance of $\omega_{j}$, $j\in\{D,R\}$\\
		$y_{_D}$ & Received signal at $D$ from $S$ in the DC scheme &
		$y_{_{R}}$ & Received signal at $R$ \\
		$y_{_{D1}}$, $y_{_{D2}}$ & Received signal at $D$ from $S$ and $R$ with relaying &
		$\gamma$ & Instantaneous SNR in $S\rightarrow D$ in the DC scheme\\
		$\gamma_{_{D1}},\gamma_{_{D2}}$ & Instantaneous SNR in $S\rightarrow D$, $R\rightarrow D$ with relaying&
		$\gamma_{_{R}}$ & Instantaneous SNR in $S\rightarrow R$\\
		$\delta$, $\delta^*$ & Delay and minimum delay&
		$\beta$, $\beta^*$ & Time sharing parameter and the optimum one\\
		$p_{\mathrm{out}}$ & Outage probability&
		$p_i$ &  Energy insufficiency probability at $i\in\{S,R\}$ \\
		$\epsilon$ &  Error probability&        
		$C(\gamma)$ & Shannon capacity for a SNR equal to $\gamma$\\
		$V(\gamma)$ & Shannon dispersion for a SNR equal to $\gamma$&
		$\hat{h}_{ij}$, $\bar{h}$ & Estimate of $h_{ij}$ and error in the estimation\\
		$x^p_i$, $y^p_j$ & Pilot symbols transmitted/received by $i/j$&
		$\gamma^{\mathrm{imp}}$ & Instantaneous SNR at the receiver with imperfect CSI\\
		$\varepsilon_{_{0}}$ & Target error probability&
		$\delta_{_{0}}$ & Maximum allowable delay\\
		$\xi$ & Error metric for the approximations in \eqref{e_fin} and \eqref{e_inf} &
		$i_{_M}+1$ & Number of terms for summation in \eqref{e_fin}\\
		\hline
	\end{tabular}
\end{table*}
\subsection{Related Work}
The most common WET techniques are based on time-switching or power splitting \cite{Zhou.2013}. Authors in \cite{Nasir.2013} propose two WET protocols called Time Switching-based Relaying (TSR) and Power Splitting-based Relaying (PSR) to be implemented by an energy constrained relay node while assisting the communications of a single link. The implementation of the TSR protocol is further analyzed later in \cite{Nasir.2015}.  
In all cases, there is not a direct link between source and destination, while only the relay is energy constrained and powered by the source. 
Alternatively, authors in \cite{Moritz.2014_1} consider a system where energy-constrained sources have independent information to transmit to a common destination, which is responsible for transferring energy wirelessly to the sources. The source nodes may cooperate with each other, under either decode-and-forward (DF) or network coding-based protocols, and even though the achievable diversity order is reduced due to wireless energy transfer process, it is very close to the one achieved for a network without energy constraints.
Those results are extended in \cite{Moritz.2014_2} by considering a more realistic consumption model where the circuitry power consumption  is taken into account. In \cite{Chen.2015}, both relay and source are powered by a WET process and the protocols: Harvest-then-Transmit (HTT) for a direct communication scenario, and Harvest-then-Cooperate (HTC) for relaying scenarios, are evaluated. However, all these works are under the ideal assumption of communicating with infinite blocklength. 

Wireless-powered communication networks at finite blocklength regime have received attention in the scientific community recently. Authors of \cite{Haghifam.2016} attain tight approximations for the outage probability/throughput in an amplify-and-forward relaying scenario, while in \cite{Makki.2016} the authors implement retransmission protocols, in both energy and information transmission phases, to reduce the outage probability compared to open-loop communications. The impact of the number of channel uses for WET and for WIT on the performance of a system where a node charged by a power beacon attempts to communicate with a receiver, is investigated in \cite{Khan.2016}. Also, subblock energy-constrained codes are investigated in \cite{Tandon.2016}, while authors provide a sufficient condition on the subblock length to avoid energy outage at the receiver.
In \cite{Lopez2.2017} we optimize a single-hop wireless system with WET in the downlink and WIT in the uplink, under quasi-static Nakagami-m fading in URC-S scenarios.
The impact of a non-cooperative dual-hop setup is further evaluated in \cite{Lopez.2017}.   
Finally and as an extension of \cite{Lopez2.2017}, the error probability and energy consumption at finite block length and finite battery energy are characterized in \cite{Lopez3_2017} for scenarios with/without energy accumulation between transmission rounds with transmit power control.

The only source of energy consumption in all above works is the transmission power, which is a very impractical simplification for energy-limited systems.
Also, all these works consider perfect CSI acquisition, which is a common situation in current scientific literature. However, the analysis under perfect CSI could be misleading for a wireless-powered communication network, due to its inherent energy constraints, and even more on systems with limited delay. Only few works have given a step forward on this direction in order to characterize different scenarios. Particularly interesting is the work in \cite{Schiessl.2016} where authors attain a closed-form approximation for the error probability at finite blocklength and imperfect CSI while showing the convenience of adapting the training sequence length and the transmission rate to reach higher reliability. However, new scenarios require deeper understanding, and battery capacity of devices and other energy consumption sources beyond transmission must be taken into account for a more realistic analysis.
\subsection{Contributions}
This paper aims at filling the above gaps in the state of the art. Next we list the main contributions of this work:
\begin{itemize}		
	\item  We model an URC-S system under Rayleigh quasi-static fading at finite blocklength and finite battery constraint with other power consumption sources beyond data transmission power. We show that the infinite battery assumption is permissible for the scenarios under discussion since the energy harvested is very low; but considering other power consumption sources, e.g., circuit and baseband processing power, is crucial because they constitute a non-negligible cause of outage.	
	\item  We analyze the impact of imperfect CSI acquisition from pilot signals on the system performance. Numerical results show the existence of an optimum pilot transmit power for channel estimation, which increases with the harvested energy and decreases with the information blocklength. We show that the energy devoted to the CSI acquisition, which depends on the power and utilized time, has to be taken into account due to the inherent energy and delay constraints of the discussed scenarios. 	
	\item For ideal direct communication, we reach accurate closed-form approximations for the error probability  for finite and infinite battery devices.
	\item Also, cooperation appears as a viable solution to meet the reliability and delay constraints of URC-S scenarios, and it is advisable transmitting with shorter blocklengths during the broadcast phase. 
\end{itemize}
The rest of this paper is organized as follows. Section \ref{system} presents the system model. Sections \ref{II} and \ref{III} discuss the delay and error probability metrics for scenarios with/without perfect CSI acquisition, respectively. Section \ref{results} presents some numerical results, while Section \ref{conclusions} concludes the paper.

\noindent{\textbf{Notation:}} Let $\mathbb{P}[A]$ denote the probability of event $A$, while $X\sim\mathrm{Exp}(1)$ is a normalized exponentially distributed random variable with Probability Density Function (PDF) $f_X(x)=e^{-x}$ and Cumulative Distribution Function (CDF) $F_X(x)=1-e^{-x}$. Then, $Y\sim\mathcal{CN}(0,1)$ is a zero-mean circularly symmetric complex random variable with unit variance. Let $\mathbb{E}[\!\ \cdot\ \!]$ denote expectation and $|\cdot|$ is the absolute value operator. The Gaussian Q-function is denoted as $Q(x)= \tfrac{1}{2} \operatorname{erfc}\left(\tfrac{x}{\sqrt{2}}\right) = \int_{x}^{\infty}\frac{1}{\sqrt{2\pi}}e^{-t^2/2}\mathrm{d}t$ \cite[\S 8.250-4]{Jeffrey.2007},  $\mathrm{Ei}(i,j)=\int_{1}^{\infty}\frac{e^{-jt}}{t^i}\mathrm{d}t$ \cite[\S 8.21]{Jeffrey.2007} is the exponential integral and $\operatorname{K}_t(\mathrel{\cdot})$  is the modified Bessel function of second kind and order $t$ \cite[\S 8.407]{Jeffrey.2007}. Finally, $\min(x,y)$ represents the minimum value between $x$ and $y$. The main symbols along the paper and their meaning are summarized in TABLE~I.
\section{System Model}\label{system}
We consider a dual-hop cooperative network where a DF relay node, $R$, is available for assisting the transmissions from the source, $S$, to the destination, $D$. All the nodes are single antenna, half-duplex devices, where $D$ is assumed to be externally powered and acts as an interrogator, e.g., an access point or a base station, requesting information from $S$, which along with $R$, may be seen as sensor nodes with very limited energy supply and finite battery of capacity $B_{_{\mathrm{max}}}$. The communication scheme is illustrated in Fig.~\ref{Fig1}. 
First, $D$ powers $S$ and $R$ during $v$ channel uses through a WET process. Right after, $S$ broadcasts a known pilot signal of $u_1$ channel uses before sending its data during a WIT phase over the next $n_1$ channel uses. The pilot is used by the receivers, $D$ and $R$, to acquire CSI. Then, $R$ tries to decode the information from $S$ and if succeeds, transmits a pilot signal of $u_2$ channel uses to $D$ for CSI acquisition and then forwards information received during the last WIT phase over the next $n_2$ channel uses. Finally, $D$ combines the received information from $S$ and $R$ to decode the message. At each WIT phase, $k$ information bits constitute the transmitted message. 
We assume $u=u_1=u_2$, and notice that an equivalent direct communication (DC) scheme can be easily obtained from Fig.~\ref{Fig1} without the assistance of $R$, with only one pilot signal transmission ($u$ channel uses) and one WIT ($n$ channel uses) phase. 
\begin{figure}[t!] 
	\centering
	\begin{subfigure}[b]{0.45\textwidth}
		\centering
		\includegraphics[width=1.05\textwidth]{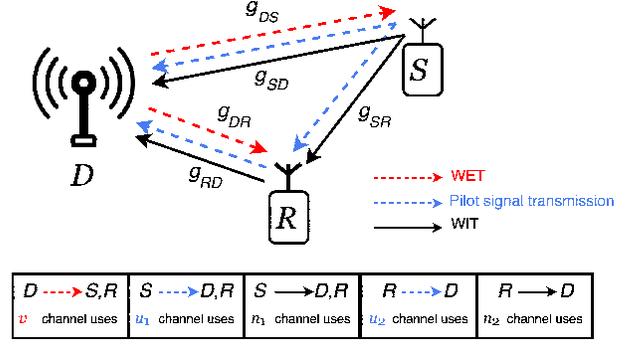}		
	\end{subfigure}    
	\caption{\small Relaying communication scheme: $D$ powers $S$ and $R$ using $v$ channel uses, then $S$ broadcasts the pilots in $u_1$ channel uses and the information into $n_1$ channel uses, while $R$ repeats the processes using, respectively, $u_2$ and  $n_2$ channel uses.	}
	\label{Fig1}	
	\vspace{-2mm}
\end{figure} 

We assume Rayleigh quasi-static fading channels, where the fading process is constant over a transmission round, which spans over all the phases of the communication scheme shown in Fig.~\ref{Fig1} ($v+u_1+u_2+n_1+n_2$ channel uses), and it is independent and identically distributed from round to round. The normalized channel gains from node $i$ to node $j$, where $i,j\in\{S,D,R\}$ and $i\ne j$, are denoted by $h_{ij}$, while $g_{ij}=h_{ij}^2\sim\mathrm{Exp}(1)$ is the power gain coefficient.
Similarly, distance between nodes is denoted by $d_{ij}$. In addition, no energy accumulation is allowed between consecutive transmission rounds, thus all the harvested energy at $S$ and $R$ is used for powering their circuits and for pilot and data transmission.
The pilot signal transmit power is assumed fixed and equal to $P_{\mathrm{csi}}$ for both $S$ and $R$. Also, the duration of each channel use is denoted by $T_c$.

\section{System Performance with Perfect CSI}\label{II}
First, as a benchmark, we assume that regardless the resources used for CSI acquisition, $u$ and $P_{\mathrm{csi}}$, receivers get a perfect channel estimation (pCSI). We note that scenarios with pCSI are quite frequent in the literature, e.g., \cite{Nasir.2013,Nasir.2015,Chen.2015,Haghifam.2016,Makki.2016,Khan.2016,Tandon.2016,Lopez2.2017,Lopez.2017}. Moreover, we start analyzing the outage probability for the DC scheme and then we extend the results for the cooperative scheme.

\subsection{DC scheme}\label{ssc:dcschme}
The energy harvested at $S$ in the WET phase, which lasts for $vT_c$, is \cite{Chen.2015}
\begin{align}
E_{_{S}}&=\frac{\eta P_{_{D}} g_{_{DS}}}{\kappa_{_{DS}}}vT_c, \label{E_s}
\end{align}
and the charging state of the battery before transmitting is
\begin{align}
B_{_{S}}&=\min(E_{_{S}},B_{_{\mathrm{max}}})=\left\{ \begin{array}{ll}
\dfrac{vT_c\eta P_{_{D}} g_{_{DS}}}{\kappa_{_{DS}}},& \mathrm{if}\ g_{_{DS}}<\lambda_{_S}\\
B_{\mathrm{max}}, & \mathrm{if}\ g_{_{DS}}\ge\lambda_{_S}
\end{array}\right.\nonumber\\
&=\dfrac{vT_c\eta P_{_{D}} \min(g_{_{DS}},\lambda_{_S})}{\kappa_{_{DS}}},
\label{charge_S}
\end{align}
since no energy is accumulated between transmission rounds and because of the finite battery at $S$, it is not allowed to store more energy than the allowable limit $B_{_{\mathrm{max}}}$. $P_{_{D}}$ is the transmit power of $D$ when wirelessly powering $S$ (and $R$), $0<\eta<1$ is the energy conversion efficiency and $\kappa_{ij}=\kappa d_{ij}^{\alpha}$ accounts for the path loss in the $i\rightarrow j$ link, where $\alpha$ is the path loss exponent, and $\kappa$ represents a combination of other losses as due to carrier frequency and the antenna gains \cite{Goldsmith.2005}. Also, $\lambda_{_S}$ is the $D\rightarrow S$ channel power gain threshold for the saturation of the battery in $S$, which is given by
\begin{align}
	\lambda_{_S}= \frac{B_{_{\mathrm{max}}}\,\kappa_{_{DS}}}{v \, T_c \, \eta P_{_{D}}}.
\end{align}
In addition, we assume that $P_{_D}$ is sufficiently large such that the energy harvested from noise is negligible.  Notice that \eqref{charge_S} is a valid expression if and only if all the energy harvested by $S$ at each round is completely used during the pilot signal transmission and WIT phases of the same round. We acknowledge that energy accumulation between transmission rounds with adequate power allocation strategies could improve the system performance, however that scenario is out of the scope of this work.

Right after the WET phase, $S$ transmits a pilot signal of $u$ channel uses and power $P_{\mathrm{csi}}$ to be used by $D$ for CSI acquisition. Also, we here take into account the circuit and baseband processing power consumption, namely $P_c$, which is assumed to be constant without loss of generality. When the harvested energy is insufficient for pilot transmission,  WIT and other consumption sources, there is an outage and transmission fails, otherwise all the remaining energy is used for those processes. Notice that an outage event happens when the transmitter has not the sufficient power to feed its circuits and send the estimating pilot or when there is data transmission but the receiver is unable to successfully decode the received message. Thus, the available energy for transmission at $S$ and its transmit power are stated, respectively, as
\begin{align}
E_{_{S}}^t&=B_{_{S}}-P_{\mathrm{csi}}uT_c-P_c(u+n)T_c, \label{E_st}\\
P_{_{S}}&=\frac{E_{_{S}}^t}{nT_c}\!
\stackrel{(a)}{=}\!\frac{v \eta P_{_{D}}\min(g_{_{DS}},\lambda_{_S})}{n\kappa_{_{DS}}}\!\!-\frac{u}{n}(P_{\mathrm{csi}}\!+\!P_c)\!-\!P_c,\label{P_s}
\end{align}
where $(a)$ in \eqref{P_s} comes from substituting \eqref{charge_S} into \eqref{E_st}. In addition, the received signal at $D$ is
\begin{align}
y_{_{D}}&=\!\sqrt{\frac{P_{_{S}}g_{_{SD}}}{\kappa_{_{SD}}}}x_{_{S}}+\omega_{_{D}},\label{yDs}
\end{align}
where $x_{_{S}}$ is the codebook transmitted by $S$, which is assumed Gaussian with zero-mean and unit-variance, $\mathbb{E}[|x_{_S}|^2]=1$. Notice that there is no guarantee that Gaussian codebooks will provide the best reliability performance at finite blocklength, and here they are merely used to gain in mathematical tractability (authors in \cite{Polyanskiy.2010} derive accurate performance approximations). 
$\omega_{_{D}}$ is the Gaussian noise vector at $D$ with variance $\sigma_{_D}^2$. 
Thus, by using \eqref{yDs} and \eqref{P_s}, we can write the instantaneous SNR at $D$ as
\begin{align}
\gamma&=\frac{v \eta P_{_{D}}g_{_{SD}}\!\min(g_{_{DS}},\lambda_{_S})}{n\kappa_{_{DS}}^{2}\sigma_{_D}^2}\!\!-\frac{g_{_{SD}}}{\kappa_{_{DS}}\sigma_{_D}^2 }\!\bigg[\frac{u}{n}(P_{\mathrm{csi}}\!+\!P_c)\!+\!P_c\bigg]\nonumber\\
&=\phi g_{_{SD}}\min(g_{_{DS}},\lambda_{_S})-\varphi g_{_{SD}},\label{gamma}
\end{align}
where
\begin{align}
\phi&=\frac{v \eta P_{_{D}}}{n\kappa_{_{DS}}^{2}\sigma_{_D}^2},\label{a}\\
\varphi&=\frac{1}{\kappa_{_{DS}}\sigma_{_D}^2 }\!\bigg[\frac{u}{n}(P_{\mathrm{csi}}\!+\!P_c)\!+\!P_c\bigg].\label{b}
\end{align}
Let $\delta$ be the delay in delivering a message of $k$ bits, while $\delta^*$ is the minimum delay that satisfies a given reliability constraint. Moreover, $\beta$ is the time sharing parameter representing the fraction of $\delta$ devoted to WET only. Therefore,
\begin{align}
\delta&=v+u+n, \label{delta}\\
\beta&=v/\delta. \label{nu}
\end{align}
Notice that $\delta$ is measured in channel uses, while $\delta T_c$ would be the delay in seconds. Finally, we define the optimum WIT blocklength, in the sense of minimizing $\delta^*$, as $n^*$. Both $\delta^*$ and $n^*$ are numerically investigated in Section \ref{results}.

On the other hand, the reliability analysis comes from evaluating the outage probability. An outage event may be due to a low harvested energy, precluding the whole WIT process, or if after the WIT process there is a decoding error at $D$. Hence, the outage probability is given as follows
\begin{equation}
p_{\mathrm{out}}=p_{_S}+\mathbb{E}[\epsilon(\gamma,k,n)],\label{pout}
\end{equation}
where $p_{_S}$ is the probability that the harvested energy at $S$ is insufficient simultaneously for channel estimation and for satisfying other consumption requirements at $S$, and it is obtained through
\begin{align}\label{ps}
p_{_S}&=\mathbb{P}[B_{_{S}}<P_{\mathrm{csi}}uT_c+P_c(u+n)T_c]\nonumber\\
&=\mathbb{P}\bigg[\dfrac{vT_c\eta P_{_{D}} \min(g_{_{DS}},\lambda_{_S})}{\kappa_{_{DS}}}<P_{\mathrm{csi}}uT_c+P_c(u+n)T_c\bigg]\nonumber\\
&=\mathbb{P}\bigg[\min(g_{_{DS}},\lambda_{_S})<\frac{\kappa_{_{DS}} }{v\eta P_{_{D}}}\Big(P_{\mathrm{csi}}u+P_c(u+n)\Big)\bigg]\nonumber\\
&\!\stackrel{(a)}{=}\!\mathbb{P}\bigg[\min(g_{_{DS}},\lambda_{_S})\!<\!\frac{\varphi}{\phi}\bigg]\!=\!1\!-\!\mathbb{P}\bigg[g_{_{DS}}\!\ge\!\frac{\varphi}{\phi}\bigg]\mathbb{P}\bigg[\lambda_{_S}\!\ge\!\frac{\varphi}{\phi}\bigg]\nonumber\\
&\stackrel{(b)}{=}\!1\!-\!\mathbb{P}\bigg[g_{_{DS}}\!\ge\!\frac{\varphi}{\phi}\bigg]=\mathbb{P}\bigg[g_{_{DS}}\!<\!\frac{\varphi}{\phi}\bigg]\nonumber\\
&\stackrel{(c)}{=}1-\exp\left({-\frac{\varphi}{\phi}}\right).
\end{align}
Notice that $(a)$ comes from using the definitions of $\phi$ and $\varphi$ given in \eqref{a} and \eqref{b}, respectively, while $(b)$ comes from the fact that $\lambda_{_S}$ is not a random variable and it should be greater than $\frac{\kappa_{_{DS}}}{v\eta P_{_{D}}}\big(P_{\mathrm{csi}}u+P_c(u+n)\big)$ for every practical system since otherwise the system is in outage all the time. Equality in $(c)$ is obtained by using the CDF expression of $g_{_{DS}}$. Also, $\epsilon(\gamma,k,n)$ is the error probability when transmitting a message of $k$ information bits over $n$ channel uses and being received with SNR equal to $\gamma$ at the destination, thus $\mathbb{E}[\epsilon(\gamma,k,n)]$ is the average probability. Both terms are accurately approximated by \eqref{Q} \cite{Polyanskiy.2010}, and \eqref{error} for quasi-static fading channel \cite{Yang.2014j}, when $n\ge100$,\footnote{See \cite[Figs.~12 and 13]{Polyanskiy.2010} and related analyses for more insight on the accuracy of \eqref{Q}. Authors show that the approximate achievable rate matches almost perfectly its true value for $n\ge 100$. Many other works, such as \cite{Makki.2016,Khan.2016,Devassy.2014,Haghifam.2016,Lopez.2017,Lopez2.2017,Lopez3_2017,Schiessl.2016,Parisa.2017,Hu.2016_2,Makki.2014}, use this accurate approximation and/or \eqref{error} to gain in tractability.} as shown next    
\begin{align}
\epsilon(\gamma,k,n)&\approx Q\Biggl(\frac{C(\gamma)-r}{\sqrt{V(\gamma)/n}}\Biggl),\label{Q}\\
\mathbb{E}[\epsilon(\gamma,k,n)]&\approx \mathbb{E}\Bigg[Q\Biggl(\frac{C(\gamma)-r}{\sqrt{V(\gamma)/n}}\Biggl)\Bigg]\nonumber\\
&\approx\int_{\gamma}Q\Biggl(\frac{C(\gamma)-r}{\sqrt{V(\gamma)/n}}\Biggl)f_{\gamma}(\gamma)\mathrm{d}\gamma,\label{error}
 \end{align}
where $r = k/n$ is the source fixed transmission rate, $C(\gamma)=\log_2(1+\gamma)$ is the Shannon capacity, $V(\gamma)=\left(1-\frac{1}{(1+\gamma)^2}\right)(\log_2e)^2$ is the channel dispersion, which measures the stochastic variability of the channel relative to a deterministic channel with the same capacity \cite{Polyanskiy.2010}.
\begin{lemma}
	The average error probability is given in \eqref{e_fin} and \eqref{e_inf} for finite and infinite battery, respectively.
		\begin{align}
		\mathbb{E}_{\mathrm{fin}}[\epsilon(\gamma,k,n)]\!&\approx\!\int\limits_{0}^{\infty}\!\int\limits_{\tfrac{\varphi}{\phi}}^{\lambda_{_S}}\!q_1(g_{_{DS}},g_{_{SD}},n)e^{\!-g_{_{DS}}}\mathrm{d}g_{_{DS}}\mathrm{d}g_{_{SD}}+\nonumber\\
		&+\!\int\limits_{0}^{\infty}\!\int\limits_{\lambda_{_S}}^{\infty}\!q_1(\lambda_{_S},g_{_{SD}},n)e^{\!-g_{_{DS}}}\mathrm{d}g_{_{DS}}\!\mathrm{d}g_{_{SD}}\!,\label{e_fin}\\
		\mathbb{E}_{\mathrm{inf}}[\epsilon(\gamma,k,n)]&\!\approx\!\int\limits_{0}^{\infty}\!\int\limits_{\tfrac{\varphi}{\phi}}^{\infty}q_1(g_{_{DS}},g_{_{SD}},n)e^{\!-g_{_{DS}}}\mathrm{d}g_{_{DS}}\!\mathrm{d}g_{_{SD}}\!, \label{e_inf}
		\end{align}
		where 
		\begin{align}
		q_1(x_1,x_2,n)&=Q\Biggl(\!\frac{C((\phi x_1\!-\!\varphi)x_2)\!-\!k/n}{\sqrt{V((\phi x_1\!-\!\varphi) x_2)/n}}\Biggl)e^{-x_2}.
		\end{align}
\end{lemma}
\begin{proof}	
	Substituting \eqref{gamma} into \eqref{error} and using the $\mathrm{PDF}$ expression for $g_{_{SD}}$ and $g_{_{DS}}$, we attain \eqref{e_fin} and \eqref{e_inf}. Notice that for the infinite battery case, \eqref{gamma} becomes $\gamma=\phi g_{_{SD}}g_{_{DS}}-\varphi g_{_{SD}}$.
\end{proof}

We consider as ideal system, one where all the energy harvested by $S$ is used only during the WIT phase ($P_{\mathrm{csi}}=P_c=0\rightarrow\varphi=0$) while, even so, $D$ has perfect knowledge of the channel.

\begin{theorem}\label{Th1}
The outage probability given in \eqref{pout} can be approximated as in \eqref{pout_fin} and \eqref{pout_inf} (on the top of the next page) for an ideal system as described in Section~\ref{ssc:dcschme}, where $\theta=2^{k/n}-1$, $\psi=\sqrt{\frac{n}{2\pi}}(2^{2k/n}-1)^{-\tfrac{1}{2}}$, $\varrho=\theta-\tfrac{1}{\psi}\sqrt{\tfrac{\pi}{2}}$, $\vartheta=\theta+\tfrac{1}{\psi}\sqrt{\tfrac{\pi}{2}}$, for finite and infinite battery devices, respectively. $i_M$ is a parameter to limit the number of terms in the summation when evaluating \eqref{pout_fin}.
\end{theorem}
\begin{proof}
	See Appendix~\ref{App_A}.\phantom\qedhere
\end{proof}
\begin{corollary}
\eqref{pout_inf} reduces to \cite[Eq.(10)]{Lopez2.2017} for Rayleigh fading.
\end{corollary}

To measure the accuracy of \eqref{pout_fin} and \eqref{pout_inf}, we evaluate the following error metric
\begin{align}
\xi=\frac{\Big|\mathbb{E}[\epsilon_{\mathrm{ex}}]-\mathbb{E}[\epsilon_{\mathrm{ap}}]\Big|}{\mathbb{E}[\epsilon_{\mathrm{ex}}]},
\end{align}
where $\mathbb{E}[\epsilon_{\mathrm{ex}}]$ is the exact error probability according to \eqref{e_fin} and \eqref{e_inf}, for finite and infinite battery, respectively, while $\mathbb{E}[\epsilon_{\mathrm{ap}}]$ is the approximated value according to \eqref{pout_fin} and \eqref{pout_inf}.
Notice that $i_M$ is a parameter that must be carefully selected when evaluating the finite battery approximation given in \eqref{pout_fin} because it establishes the quantity of elements taken into account for the summation. 
As $i_M\rightarrow\infty$, expression \eqref{pout_fin} becomes more accurate but harder to evaluate and computationally heavier. Besides, a relatively small value conduces to an inaccurate approximation.
\begin{figure}[t!] 
	\centering
	\begin{subfigure}[b]{0.45\textwidth}
		\centering
		\includegraphics[width=0.95\textwidth]{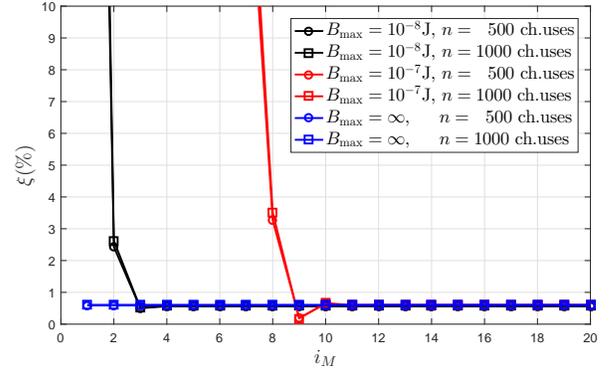}		
	\end{subfigure}    		
	\caption{\small $\xi(\%)$ as a function of $i_M$ with $v=1000$ channel uses.}\label{Fig2}	
	\vspace{-3mm}	
\end{figure}

\begin{figure*}[!t]
	\begin{align}
	&\mathbb{E}_{_{\mathrm{fin}}}[\epsilon(\gamma,k,n)]\approx 1\!-\!\bigg(\frac{1}{2}\!-\!\frac{\psi}{\sqrt{2\pi}}(\theta\!-\!\varrho\!-\!\lambda_{_S}\phi)\bigg)e^{\!-\!\lambda_{_S}\!-\!\frac{\varrho}{\lambda_{_S}\phi}}\!-\!\bigg(\frac{1}{2}\!+\!\frac{\psi}{\sqrt{2\pi}}(\theta\!-\!\vartheta\!-\!\lambda_{_S}\phi)\bigg)e^{\!-\!\lambda_{_S}\!-\!\frac{\vartheta}{\lambda_{_S}\phi}}+\sum_{i=0}^{i_M}\frac{(-1)^i\lambda_{_S}^{i+1}}{i!}\cdot\nonumber\\
	&\qquad\cdot\Bigg[\frac{\psi\lambda_{_S}\phi}{\sqrt{2\pi}}\bigg(\mathrm{Ei}\big(i\!+\!3,\tfrac{\vartheta}{\lambda_{_S}\phi}\big)-\mathrm{Ei}\big(i\!+\!3,\tfrac{\varrho}{\lambda_{_S}\phi}\big)\bigg)\!-\!\bigg(\frac{1}{2}\!-\!\frac{\psi(\theta\!-\!\varrho)}{\sqrt{2\pi}}\bigg)\mathrm{Ei}\big(i\!+\!2,\tfrac{\varrho}{\lambda_{_S}\phi}\big)\!-\!\bigg(\frac{1}{2}\!+\!\frac{\psi(\theta\!-\!\vartheta)}{\sqrt{2\pi}}\bigg)\mathrm{Ei}\big(i\!+\!2,\tfrac{\vartheta}{\lambda_{_S}\phi}\big)\Bigg],\label{pout_fin}\\
	&\mathbb{E}_{_{\mathrm{inf}}}[\epsilon(\gamma,k,n)]\approx 1-\sqrt{\frac{\varrho}{\phi}}\biggl[1+\bigg(\frac{2\psi}{\sqrt{2\pi}}\bigg)\big(\varrho+\phi-\theta\big)\biggl]K_1\bigg(2\sqrt{\frac{\varrho}{\phi}}\bigg)-\frac{2\psi\varrho}{\sqrt{2\pi}}K_0\bigg(2\sqrt{\frac{\varrho}{\phi}}\bigg)+\nonumber\\
	&\ \ \ \ \ \ \ \ \ \ \ \ \ \ \ \ \ \ \ \ \ \ \ \ \ \ \ \ \ \ \ \ \ \ \ \ \ \ \ \ \ \ \ \ \ \ \ \ \ \ \ \ \ \ \ \ -\sqrt{\frac{\vartheta}{\phi}}\biggl[1-\bigg(\frac{2\psi}{\sqrt{2\pi}}\bigg)\big(\vartheta+\phi-\theta\big)\biggl]K_1\bigg(2\sqrt{\frac{\vartheta}{\phi}}\bigg)+\frac{2\psi\vartheta}{\sqrt{2\pi}}K_0\bigg(2\sqrt{\frac{\vartheta}{\phi}}\bigg)\label{pout_inf}.  
	\end{align}
	\hrule
\end{figure*}

Fig.~\ref{Fig2} shows numerically the impact of $i_M$ on $\xi$, when $v=1000$ and $n\in\{500,1000\}$ channel uses, and $B_{_{\mathrm{max}}}\in\{10^{-8},10^{-7},\infty\}$J. The remaining system parameters were the ones chosen in Section~\ref{results}.
The larger the battery capacity, the larger the required $i_M$ for a good approximation using \eqref{e_fin}. 
Note that for all cases $i_M\ge10$ provides a good accuracy with $\xi<1\%$. The accuracy of \eqref{e_fin} and \eqref{e_inf} was also measured for many other different setups, and in all the cases we reached similar results, e.g. $\xi<3.5\%$ for any $(n,v)$ pair with $n,v\in[100\ 5000]$ and $k>64$ bits. The accuracy is good because  we used the linearization (first order approximation) of \eqref{Q} to attain \eqref{pout_fin} and \eqref{pout_inf}, which is symmetrical with respect to $\gamma=2^r-1$, and lies beneath and above of \eqref{Q} for $\gamma<2^r-1$ and $\gamma>2^r-1$, respectively. Thus, when integrating over all the channel realizations, the error tends to vanish.
\subsection{Relaying scheme}\label{DF}
Cooperative technique has rekindled enormous interests from the wireless communication community over the past decade. As shown in \cite{Parisa.2017}, the attained spatial diversity can improve communication reliability, thus it can also reduce the system delay for a target error constraint. 
For an energy-constrained URC system, cooperation seems advisable to reduce the instantaneous consumption power of devices while meeting the reliability/delay constraints.
Consequently, herein we consider the presence of node $R$, which is willing to assist the $S\rightarrow D$ communication all the time, while operating under the DF protocol. 
Under the finite blocklength regime, decoding errors may
occur. We assume that $R$ reliably detects the errors, and consequently it does not forward the message to $D$ when an error occurs, as in \cite{Hu.2016_2}. In that condition, an outage event is declared for the $S\rightarrow R$ link, and 
$R$ is inactive during the time reserved for its transmission, e.g., the last $u+n_2$ channel uses according to Fig.~\ref{Fig1}. The energy that could be saved in those cases is out of the scope of our analysis\footnote{Notice that in URC-S scenarios the probability of such events should be kept low, specially if relaying is crucial in achieving the ultra-reliability, and the impact of that remaining energy on the system performance can be negligible, as shown in \cite{Lopez.2017} for a system without a direct link.}.
This behavior reduces the system complexity compared to other works, e.g. \cite{Nasir.2015} for infinite blocklength and preset power relay.
Nodes $S$ and $R$ are assumed to have the same characteristics, e.g., battery capacity, pilot power, power consumption profile.

Similar than in \eqref{E_s} and \eqref{charge_S}, the energy harvested at $R$ and the battery charge before transmitting are, respectively, 
\begin{align}
E_{_{R}}&=\frac{\eta P_{_D}g_{_{DR}}}{\kappa_{_{DR}}}vT_c,\\
B_{_R}&=\min(E_{_R},B_{_{\mathrm{max}}})=\left\{ \begin{array}{ll}
\dfrac{vT_c\eta P_{_D} g_{_{DR}}}{\kappa_{_{DR}}},& \mathrm{if}\ g_{_{DR}}<\lambda_{_R}\\
B_{\mathrm{max}}, & \mathrm{if}\ g_{_{DR}}\ge\lambda_{_R}
\end{array}\right.\nonumber\\
&=\dfrac{vT_c\eta P_{_D} \min(g_{_{DR}},\lambda_{_R})}{\kappa_{_{DR}}},
\end{align}
where $\lambda_{_R}$ is the $D\rightarrow R$ channel power gain threshold for $R$ battery saturation given by
\begin{align}\label{lambdar}
\lambda_{_R}=\frac{B_{\mathrm{\max}}\kappa_{_{DR}}}{vT_c\eta P_{_D}}.
\end{align}
The expressions for the energy harvested at $S$, the charge of its battery, the available energy for transmission and its transmit power are the same as in the analysis of the DC scheme, respectively \eqref{E_s}, \eqref{charge_S}, \eqref{E_st} and \eqref{P_s}, but with $n=n_1$. Now, for $R$, they are given by
\begin{align}
E_{_{R}}^t&=B_{_{R}}-P_{\mathrm{csi}}uT_c-P_c\, \nu \, T_c, \label{E_rt}\\
P_{_{R}}&\!=\!\frac{E_{_{R}}^t}{n_2T_c}\!=\!\frac{v \eta P_{_{D}}\!\min(\!g_{_{DR}},\lambda_{_R}\!)}{n_2\kappa_{_{DR}}}\!-\!\frac{1}{n_2}\!\Big(uP_{\mathrm{csi}}\!+\! \nu\, P_c \!\Big),\label{P_r}
\end{align}
where $\nu = 2u+n_1+n_2$ (note that $R$ consumes circuit power for receiving the pilots and the data from $S$, besides the power consumed in the transmission of pilots and data to $D$). During the first WIT phase, $S$ broadcasts its data to $D$ and $R$. The expression of the signal received at $D$ in this phase is $y_{_{D1}}=y_{_D}$, equal to \eqref{yDs} but with $n=n_1$, and the signal received at $R$ is
\begin{align}
y_{_R}&=\sqrt{\frac{P_{_S}g_{_{SR}}}{\kappa_{_{SR}}}}x_{_S}+\omega_{_R},\label{yRC}
\end{align}
where $P_{_S}$ is given in \eqref{P_s} and $\omega_{_R}$ is the Gaussian noise vector at $R$ with variance $\sigma_{_R}^2$. When $R$ successfully decodes  the message during the first WIT phase, it re-encodes it in $n_2$ channel uses, and after the pilot signal transmission, it transmits the message to $D$ during the second WIT phase. The received signal at $D$ at this phase is thus given by 
\begin{align}
y_{_{D2}}&=\sqrt{\frac{P_{_R}g_{_{RD}}}{\kappa_{_{RD}}}}x_{_R}+\omega_{_D},\label{yD2}
\end{align}
where $P{_R}$ is given in \eqref{P_r} and $x_{_{R}}$ is the zero-mean, unit-variance Gaussian codebook transmitted by $R$. $R$ uses the same codebook, which is defined a priori, and when $n_1=n_2 \rightarrow x_{_{R}}=x_{_{S}}$.

Let $\gamma_{_{D1}}$ and $\gamma_{_{R}}$ be the instantaneous SNRs at $D$ and $R$, respectively, for the signal transmitted by $S$ during the first WIT phase. Also, let $\gamma_{_{D2}}$ be the instantaneous SNR at $D$ for the signal transmitted by $R$ during the second WIT phase, as long as $R$ achieved a successful decoding of the message transmitted by $S$. Expression \eqref{gamma} is still useful for $\gamma_{_{D1}}$ but with $n=n_1$ ($\phi_{_{D1}}=\phi|_{_{n=n_1}}$ and $\varphi_{_{D1}}=\varphi|_{_{n=n_1}}$ according to \eqref{a} and \eqref{b}, respectively) since now the transmission time depends on $n_1$. Using \eqref{yRC} and \eqref{yD2} we attain the expressions for  $\gamma_{_{R}}$ and $\gamma_{_{D2}}$, which are
\begin{align}
\gamma_{_R}&\!=\!\frac{\eta vP_{_D}g_{_{SR}}\!\min(g_{_{DS}},\lambda_{_S})}{n_1\kappa_{_{DS}}\kappa_{_{SR}}\sigma_{_R}^2}\!-\!\frac{g_{_{SR}}}{\kappa_{_{SR}}\sigma_{_R}^2}\!\bigg(\frac{u}{n_1}(P_{\mathrm{csi}}\!+\!P_c)\!+\!P_c\bigg)\nonumber\\
&=\!\phi_{_R} g_{_{SR}}\min(g_{_{DS}},\lambda_{_S})\!-\!\varphi_{_R} g_{_{SR}},\label{SNR_r}
\end{align}
\begin{align}
\gamma_{_{D2}}&\!=\!\frac{v\eta P_{_D} g_{_{RD}} \min(g_{_{DR}},\lambda_{_R})}{n_2\kappa_{_{DR}}^{2}\sigma_{_D}^2}\!-\! \frac{g_{_{RD}}}{\kappa_{_{RD}}n_2\sigma_{_D}^2}\big(uP_{\mathrm{csi}}\!+\!\nu P_c\big)\nonumber\\
&=\!\phi_{_{D2}} g_{_{RD}}\min(g_{_{DR}},\lambda_{_{R}})\!-\!\varphi_{_{D2}} g_{_{RD}},\label{SNR_d2}
\end{align}
where
\begin{align}
\label{alphar1}
\phi_{_{R}}&=\frac{\eta vP_{_D}}{n_1\kappa_{_{DS}}\kappa_{_{SR}}^{\alpha}\sigma_{_R}^2}, \\ 
\label{alphad2}
\phi_{_{D2}}&=\frac{\eta vP_{_D}}{n_2\kappa_{_{DR}}^{2}\sigma_{_D}^2},\\ 
\label{betar1}
\varphi_{_{R}}&=\frac{1}{\kappa_{_{SR}}\sigma_{_R}^2}\bigg(\frac{u}{n_1}(P_{\mathrm{csi}}+P_c)+P_c\bigg), \\
\label{betad2}
\varphi_{_{D2}}&=\frac{1}{\kappa_{_{RD}}n_2\sigma_{_D}^2}\Big(uP_{\mathrm{csi}}+ \nu P_c\Big).
\end{align}

The delay for the cooperative scheme is 
\begin{align}
\delta= v + \nu=v+2u+n_1+n_2, \label{delta_r}
\end{align}
while the time sharing parameter still obeys \eqref{nu}, but using \eqref{delta_r} for the delay. An outage event for the cooperative scheme can be due to: i) a low harvested energy at $S$, precluding the whole transmission process; ii) communication error in the $S\rightarrow D$ link at the same time that $R$ is unable to perform a transmission due to its low harvested energy or when it fails in decoding the information from $S$; iii) both $S$ and $R$ were capable of completing their transmissions ($R$ succeeded in decoding the message from $S$) but after combining both signals at $D$ there is still a decoding error. Thus, the outage probability can be mathematically written as\footnote{Note that in \eqref{outCC} and \eqref{pout} when a complementary event, e.g. $1-p_{_S}$ or $1-p_{_R}$, does not appear explicitly in the equation is because its effect is implicit when evaluating the communication errors by properly setting the integration limits.} 
\begin{align}\label{outCC}
p_{_{\mathrm{out}}}=p_{_S}+ &p_{_R}\mathbb{E}[\epsilon(\gamma_{_{D1}},k,n_1)]\!+\!\mathbb{E}\Big[\big(1\!-\!\epsilon(\gamma_{_R},k,n_1)\big)\epsilon_{\mathrm{comb}}\Big]\!+\!\nonumber\\
&+(1\!-\!p_{_R})\mathbb{E}\Big[\epsilon(\gamma_{_R},k,n_1) \epsilon(\gamma_{_{D1}},k,n_1)\Big],
\end{align}
where $p_{_S}$ can be calculated as in \eqref{ps} but with $n=n_1$, and $p_{_R}$ is the probability that the harvested energy at $R$ is insufficient simultaneously for channel estimation and for satisfying other consumption requirements at $R$, and can be easily obtained similarly to $p_{_S}$ as 
\begin{align}
p_{_R}&=\mathbb{P}\big[B_{_R}<P_{\mathrm{csi}}\, u\, T_c+P_c \,\nu\, T_c\big]\nonumber\\
&=\!1\!-\!\exp\left({-\frac{\kappa_{_{DR}}}{v \eta P_{_D}}\big(P_{\mathrm{csi}}\, u+P_c\, \nu \big)}\right)\nonumber\\
&=1\!-\!\exp\left({-\frac{\varphi_{_{D2}}}{\phi_{_{D2}}}}\right).
\end{align}
$\mathbb{E}[\epsilon(\gamma_{_{D1}},k,n_1)]$ can be obtained through \eqref{e_fin} (or \eqref{e_inf} for the infinite battery case) with $\gamma=\gamma_{_{D1}}$, $n=n_1$. Also, $\epsilon_{\mathrm{comb}}\in\{\epsilon_{\mathrm{sc}},\epsilon_{\mathrm{mrc}}\}$ is the error probability when $D$ tries to decode the information from the combination of signals received from both WIT phases and depends on the used combination technique: Selection Combining ($\mathrm{SC}$) or Maximal Ratio Combining ($\mathrm{MRC}$).
Notice that
\begin{align}
\mathbb{E}\Big[\epsilon(\gamma_{_R},k,n_1)\epsilon(\gamma_{_{D1}},k,n_1)\Big]&\!\ne\!\mathbb{E}[\epsilon(\gamma_{_R},k,n_1)]\mathbb{E}[\epsilon(\gamma_{_{D1}},k,n_1)],\nonumber\\
\mathbb{E}\big[\big(1\!-\!\epsilon(\gamma_{_R},k,n_1)\epsilon_{\mathrm{comb}}\big)\big]&\!\ne\!\mathbb{E}\big[\big(1\!-\!\epsilon(\gamma_{_R},k,n_1)\big)\big]\mathbb{E}[\epsilon_{\mathrm{comb}}],\nonumber
\end{align}  
since $\gamma_{_R}$, $\gamma_{_{D1}}$ and $\gamma_{\mathrm{comb}}$ are correlated through the variable $g_{_{DS}}$. In the first case, the expectation can be evaluated through \eqref{E1} and \eqref{E2} for finite and infinite battery, respectively, as shown next
\begin{align}
	&\mathbb{E}_{\mathrm{fin}}\Big[\epsilon(\gamma_{_R},k,n_1)\epsilon(\gamma_{_{D1}},k,n_1)\Big]\nonumber\\
	&\!\approx\!\int\limits_{0}^{\infty}\!\int\limits_{0}^{\infty}\!\int\limits_{\tfrac{\varphi}{\phi}}^{\lambda_{_S}}\!q_1(g_{_{DS}},g_{_{SD}},n_1)q_2(g_{_{DS}},g_{_{SR}})e^{\!-\!g_{_{DS}}}\mathrm{d}g_{_{DS}}\mathrm{d}g_{_{SD}}\mathrm{d}g_{_{SR}}\!+\!\nonumber\\
	&\!+\!\!\int\limits_{0}^{\infty}\!\!\int\limits_{0}^{\infty}\!\!\int\limits_{\lambda_{_S}}^{\infty}\!\!q_1(\lambda_{_S}\!,g_{_{SD}},n_1\!)q_2(\lambda_{_S}\!,g_{_{SR}}\!)e^{\!-\!g_{_{DS}}}\!\mathrm{d}g_{_{DS}}\!\mathrm{d}g_{_{SD}}\!\mathrm{d}g_{_{SR}},\label{E1}\\
	&\mathbb{E}_{\mathrm{inf}}\Big[\epsilon(\gamma_{_R},k,n_1)\epsilon(\gamma_{_{D1}},k,n_1)\Big]\nonumber\\
	&\!\!\approx\!\!\int\limits_{0}^{\infty}\!\!\int\limits_{0}^{\infty}\!\!\int\limits_{\tfrac{\varphi}{\phi}}^{\infty}\!\!q_1(g_{_{DS}},g_{_{SD}},n_1\!)q_2(g_{_{DS}},g_{_{SR}})e^{\!-\!g_{_{DS}}}\!\mathrm{d}g_{_{DS}}\!\mathrm{d}g_{_{SD}}\!\mathrm{d}g_{_{SR}}\!, \label{E2}
\end{align}
where
\begin{align}
q_2(x_1,x_2)=Q\Biggl(\!\frac{C((\phi_{_R} x_1\!-\!\varphi_{_R})x_2)\!-\!k/n_1}{\sqrt{V((\phi_{_R} x_1\!-\!\varphi_{_R}) x_2)/n_1}}\Biggl)e^{-x_2}.
\end{align}

The classical $\mathrm{SC}$ technique establishes that the signal with the highest SNR is selected since its error probability is the lowest \cite{Goldsmith.2005}. Nonetheless, when the blocklength is small and different for each link, $n_1 \ne n_2$, this assumption does not hold always as is induced from \eqref{Q}, which motivates us to state the outage probability of the SC combining scheme as
\begin{align}\label{SC}
\epsilon_{\mathrm{sc}}=\min\big(\epsilon(\gamma_{_{D1}},k,n_1),\epsilon(\gamma_{_{D2}},k,n_2)\big),
\end{align}
where $\epsilon(\gamma_{_{D1}},k,n_1)$ and $\epsilon(\gamma_{_{D2}},k,n_2)$ can be obtained through \eqref{Q} for each channel realization. Thus, $D$ has to select the signal to be decoded based on 
\begin{align}\label{SC_new}
\frac{C(\gamma_{_{D1}})-k/n_1}{\sqrt{V(\gamma_{_{D1}})/n_1}}\stacksign{y_{_{D1}}}{y_{_{D2}}}\frac{C(\gamma_{_{D2}})-k/n_2}{\sqrt{V(\gamma_{_{D2}})/n_2}}.
\end{align}
Notice that this is the optimum strategy when $D$ selects only one signal; although it can be difficult to implement. In addition, it seems intractable to find a closed-form expression for \eqref{SC}, that is why in Section~\ref{results} we estimate it numerically. 

On the other hand, the $\mathrm{MRC}$ technique reaches better results than $\mathrm{SC}$, although with higher hardware complexity \cite{Goldsmith.2005}, and limited to the case where $n_1=n_2$. The error probability for that case is given by
\begin{align}\label{mrc}
\epsilon_{\mathrm{mrc}}=\epsilon(\gamma_{_{D1}}+\gamma_{_{D2}},k,n_1),
\end{align}
\begin{figure*}[!t]		
	\begin{align}
	\mathbb{E}_{_{\mathrm{fin}}}\Big[\big(1-\epsilon(\gamma_{_R},k,n_1)\big)&\epsilon_{_{\mathrm{MRC}}}\Big]	
	\approx\int\limits_{0}^{\infty}\int\limits_{0}^{\infty}\int\limits_{0}^{\infty}\int\limits_{\tfrac{\varphi_{_{D1}}}{\phi_{_{D1}}}}^{\lambda_{_S}}\int\limits_{\tfrac{\varphi_{_{D2}}}{\phi_{_{D2}}}}^{\lambda_{_R}}q_3(g_{_{DR}},g_{_{DS}},g_{_{RD}},g_{_{SD}},g_{_{SR}})e^{-g_{_{DS}}-g_{_{DR}}}\mathrm{d}g_{_{DR}}\mathrm{d}g_{_{DS}}\mathrm{d}g_{_{RD}}\mathrm{d}g_{_{SD}}\mathrm{d}g_{_{SR}}+\nonumber\\
	&\qquad+\int\limits_{0}^{\infty}\int\limits_{0}^{\infty}\int\limits_{0}^{\infty}\int\limits_{\lambda_{_S}}^{\infty}\int\limits_{\lambda_{_R}}^{\infty}q_3(\lambda_{_R},\lambda_{_S},g_{_{RD}},g_{_{SD}},g_{_{SR}})e^{-g_{_{DS}}-g_{_{DR}}}\mathrm{d}g_{_{DR}}\mathrm{d}g_{_{DS}}\mathrm{d}g_{_{RD}}\mathrm{d}g_{_{SD}}\mathrm{d}g_{_{SR}} \label{mrc_fin}\\
	\mathbb{E}_{_{\mathrm{inf}}}\!\Big[\big(1\!-\!\epsilon(\gamma_{_R}\!,k,n_1)\big)&\epsilon_{_{\mathrm{MRC}}}\!\Big]\!\approx\!\int\limits_{0}^{\infty}\!\int\limits_{0}^{\infty}\!\int\limits_{0}^{\infty}\!\int\limits_{\tfrac{\varphi_{_{D1}}}{\phi_{_{D1}}}}^{\infty}\!\int\limits_{\tfrac{\varphi_{_{D2}}}{\phi_{_{D2}}}}^{\infty}\!\!q_3(g_{_{DR}},g_{_{DS}},g_{_{RD}},g_{_{SD}},\!g_{_{SR}})e^{-g_{_{DS}}-g_{_{DR}}}\mathrm{d}g_{_{DR}}\mathrm{d}g_{_{DS}}\mathrm{d}g_{_{RD}}\mathrm{d}g_{_{SD}}\mathrm{d}g_{_{SR}}\label{mrc_inf}
	\end{align}
	\begin{align}
	&q_3(x_1,x_2,x_3,x_4,x_5)\!=\!\Biggl(\!1\!-\!Q\bigg(\frac{C(x_2(\phi_{_{_R}}x_5\!-\!\varphi_{_{_R}}))\!-\!k/n_1}{\sqrt{V(x_2(\phi_{_{_R}}x_5\!-\!\varphi_{_{_R}}))/n_1}}\bigg)\!\Biggl)\!Q\!\Biggl(\!\frac{C(x_4(\phi_{_{D1}}x_2\!-\!\varphi_{_{D1}})\!+\!x_3(\phi_{_{D2}} x_1\!-\!\varphi_{_{D2}}))\!-\!k/n_1}{\sqrt{V(x_4(\phi_{_{D1}}x_2-\varphi_{_{D1}})+x_3(\phi_{_{D2}} x_1\!-\!\varphi_{_{D2}}))/n_1}}\!\Biggl)\!e^{-\!x_3\!-\!x_4\!-\!x_5}\label{q}
	\end{align}
	\hrule
\end{figure*}

\begin{figure}[t!] 
	\centering
	\begin{subfigure}[b]{0.45\textwidth}
		\centering
		\includegraphics[width=0.95\textwidth]{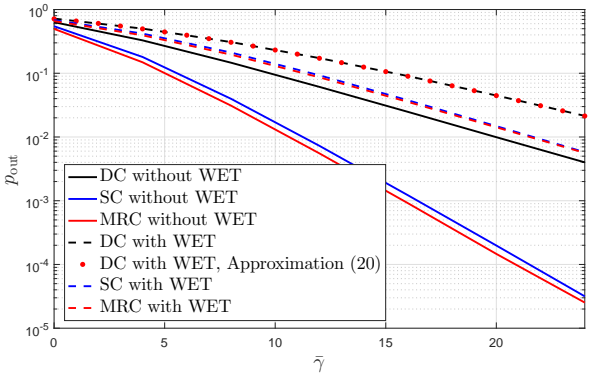}
	\end{subfigure}    
	\vspace{-1mm}		
	\caption{\small SC vs MRC for and ideal system with equal mean SNR ($\bar{\gamma}$) in all links and equal blocklengths in the broadcast and cooperation phases, for the case of without WET and with WET. For scenarios with WET, $\bar{\gamma}=\phi=\phi_{_{D1}}=\phi_{_{D2}}=\phi_{_{R}}$.}\label{Fig3}
	\vspace{-3mm}
\end{figure}
\noindent and the last term in \eqref{outCC} for finite and infinite battery is given in \eqref{mrc_fin} and \eqref{mrc_inf} (on the top of the next page), respectively, where $q_3(x_1,x_2,x_3,x_4,x_5)$ is defined in \eqref{q}.

\subsection{Impact of the Combining Technique}\label{mrc_sc}
As we mentioned above, the MRC technique reaches better results than SC for a conventional network without WET \cite{Goldsmith.2005}. Notice that without WET and assuming that nodes $S$ and $R$ transmit the information with fixed power, the links $S\rightarrow D$ and $S\rightarrow R$ are completely independent. However, for the model with WET analyzed in this work, those links are correlated through channel coefficient $g_{_{DS}}$ and therefore, the diversity cannot be fully achieved. 
Based on that, it is not surprising that the impact of the combining scheme used at $D$ is much lower for the case with WET than for a conventional network as shown in Fig.~\ref{Fig3}. As shown in the figure, there is not a significant performance difference between SC and MRC, being practically indistinguishable for $p_{\mathrm{out}}<10^{-2}$, which is our study case since we are analyzing URC-S scenarios. Also, notice that the performance gap between DC and the cooperative scheme is much smaller for the case with WET. Nonetheless, using cooperation remains advantageous and further discussion takes place in Section~\ref{results}. Additionally, Fig.~\ref{Fig3} also corroborates the accuracy claimed when giving \eqref{pout_inf}.
\section{System Performance with Imperfect CSI}\label{III}
Herein we consider a more realistic case where the channel estimate $\hat{h}_{ij}$ made by node $j\in\{R, D\}$ of the true channel coefficient $h_{ij}$ is imperfect (iCSI), {\it i.e.}, $h_{ij}\ne\hat{h}_{ij}$. The pilot signal received by $j$ due to the transmission of $i$ is 
\begin{align}
y_{_j}^p&=\sqrt{\dfrac{P_{\mathrm{csi}}}{\kappa_{_{ij}}}}h_{ij}x_{_i}^p+\omega_{j},
\end{align}
where $i\in\{S,R\}$ with $i\ne j$, and $x_{i}^p$ are the $u$ pilot symbols transmitted by $i$ with $\mathbb{E}[|x_{i}^p|^2]=1$.
Supposing minimum mean square error (MMSE) estimation at the receivers, the relation between the true channel coefficients, the estimates, and the estimation error can be modeled as~\cite{Gursoy.2009}
\begin{align}\label{ht}
h_{ij}=\hat{h}_{ij}+\tilde{h}_{ij},
\end{align}
where $h_{ij}\sim\mathcal{CN}(0,1)$, 
\begin{align}
\hat{h}_{ij}&\sim\mathcal{CN}\Bigg(0,\dfrac{P_{\mathrm{csi}}u\kappa_{ij}^{-1}}{P_{\mathrm{csi}}u\kappa_{ij}^{-1}+\sigma_{_j}^2}\Bigg),\label{h1}
\end{align}
and $\tilde{h}_{ij}$, the error in the channel estimation, is
\begin{align}
\tilde{h}_{ij}&\sim\mathcal{CN}\Bigg(0,\dfrac{\sigma_{_j}^2}{P_{\mathrm{csi}}u\kappa_{ij}^{-1}+\sigma_{_j}^2}\Bigg).\label{h2}
\end{align}
Using \eqref{ht}, \eqref{h1} and \eqref{h2}, the signal received at $j$ from $i$ during a WIT phase is
\begin{align}
y_{_j}&=\sqrt{\dfrac{P_{_i}}{\kappa_{ij}}}h_{ij}x_{_i}+\omega_{_j}=\sqrt{\dfrac{P_{_i}}{\kappa_{ij}}}\hat{h}_{ij}x_{_i}+\sqrt{\dfrac{P_{_i}}{\kappa_{ij}}}\tilde{h}_{ij}x_{_i}+\omega_{_j}\nonumber\\
&=\sqrt{\dfrac{P_{_i}}{\kappa_{ij}}}\hat{h}_{ij}x_{_i}+w_{\mathrm{eq}},\label{yii}
\end{align}
where $w_{\mathrm{eq}}$ is the effective noise due to both the estimation error and AWGN. That noise is neither Gaussian nor independent on the data signal in general. However, it is been shown in \cite{Medard.2000} that
treating it as a circularly symmetric zero mean complex Gaussian process with variance $\sigma_{_{\mathrm{eq}}}^2=P_{_i}\kappa_{ij}^{-1}\sigma_{\tilde{h}_{ij}}^2+\sigma_{_j}^2$ provides the worst case scenario for the channel capacity (more precisely, for the mutual information between a Gaussian channel input and
the channel output)\footnote{This result has been widely used, e.g., \cite{Javaher.2010,Gursoy.2009}.}. This implies that the equivalent SNR under such assumption is a lower bound of the actual instantaneous SNR, thus (and based on \eqref{Q} and \eqref{error}), this assumption also provides the worst case scenario for the system performance  in terms of error probability. Based on \eqref{yii}, the equivalent instantaneous SNR at $j$ in the iCSI case is thus given by
\begin{align}
\gamma^{\mathrm{imp}}&=\frac{P_{_i}|h_{ij}|^2\kappa_{ij}^{-1}}{\sigma_{_{\mathrm{eq}}}^2}=\frac{P_{_i}\hat{g}_{ij}\kappa_{ij}^{-1}}{P_{_i}\kappa_{ij}^{-1}\sigma_{\tilde{h}_{ij}}^2+\sigma_{_j}^2}. \label{SNRi} 
\end{align}
\subsection{DC scheme}\label{Imp_DC}
For the DC scheme, substituting \eqref{P_s}, and $\sigma_{\tilde{h}_{_{SD}}}^2=\sigma_{_D}^2/\Big(P_{\mathrm{csi}}u\kappa_{_{SD}}^{-1}+\sigma_{_D}^2\Big)$ from \eqref{h2}, into \eqref{SNRi}, we attain the following equivalent instantaneous SNR at $D$
\begin{figure*}[!t]	
	\begin{align}
	\frac{\partial^2\gamma^{\mathrm{imp}}}{\partial P_{\mathrm{csi}}^2}\!=\!-\!2g_{_{SD}}nu^2\frac{(\kappa_{\cdot}\sigma^2\!-\! P_c(n\!+\!u)\!+\!\min(g_{_{DS}},\!\lambda)\frac{\eta P_{_{D}}v}{\kappa_{\cdot}})(\kappa_{\cdot}\sigma^2n\!-\! P_c(n\!+\!u)\!+\!\min(g_{_{DS}},\!\lambda)\frac{\eta P_{_{D}}v}{\kappa_{\cdot}})}{\sigma^2\big(\kappa_{\cdot}n\sigma^2\!-\! P_c(n\!+\!u)\!+\!\min(g_{_{DS}},\lambda)\frac{\eta P_{_{D}}v}{\kappa_{\cdot}}+(n-1)uP_{\mathrm{csi}}\big)^3}\label{gamImp_d2}\tag{54}
	\end{align}
	\hrule
\end{figure*}
\begin{align}
\gamma^{\mathrm{imp}}
&=\frac{P_{_S}g_{_{SD}}}{\kappa_{_{SD}}\sigma_{_D}^2}\bigg(\frac{P_{\mathrm{csi}}u}{P_{\mathrm{csi}}u\!+\!P_{_S}\!+\!\kappa_{_{SD}}\sigma_{_D}^2}\bigg)\nonumber\\
&=\!(\phi \min(g_{_{DS}},\lambda_{_S})\!-\!\varphi) 
\cdot\nonumber\\
&\cdot g_{_{SD}}\frac{P_{\mathrm{csi}}u}{P_{\mathrm{csi}}u+\kappa_{_{SD}}\sigma_{_D}^2(\phi\min(g_{_{DS}},\lambda_{_S})-\varphi+1)}.\label{SNRdc}
\end{align}
Thus, the effect of imperfect channel estimation can be seen as a decrease in the instantaneous SNR seeing at $D$. Thus, to find the information error probability it is just necessary to evaluate \eqref{e_fin} and \eqref{e_inf} for finite and infinite battery, respectively, but with $\gamma=\gamma^{\mathrm{imp}}$.

Substituting \eqref{a} and \eqref{b} into \eqref{SNRdc} and calculating its second derivative yields \eqref{gamImp_d2} at the top of the next page,
where subscripts of $\sigma_{_D}^2,\ \lambda_{_S}$ are omitted to shorten the notation, and  $\kappa_{\cdot}=\kappa_{_{DS}}=\kappa_{_{SD}}$ for convenience. 
We know that $B_{_S}>P_c(u+n)T_c$ is the condition required for the CSI estimation phase to take place; otherwise there is an outage due to the insufficiency energy for transmission, which is counted in the term $p_{_S}$ of \eqref{pout}. Based on \eqref{charge_S}, that equality is equivalent to $\min(g_{_{DS}},\!\lambda)\frac{\eta P_{_{D}}v}{\kappa_{\cdot}}>P_c(n\!+\!u)$, thus the fraction in \eqref{gamImp_d2} is positive and $\frac{\partial^2\gamma^{\mathrm{imp}}}{\partial P_{\mathrm{csi}}^2}<0$. Therefore $\gamma^{\mathrm{imp}}$ is concave on $P_{\mathrm{csi}}$, and since $\epsilon(\gamma^{\mathrm{imp}},k,n)$ is decreasing on $\gamma^{\mathrm{imp}}$ (the greater the SNR, the smaller the chance of error), there is only one value of $P_{\mathrm{csi}}$ that minimizes the error probability by maximizing the SNR for given channel realizations $g_{_{DS}}$ and $g_{_{SD}}$. However, when averaging over all channel realizations it becomes intractable to prove mathematically the existence of a unique optimum pilot transmit power, let alone find it, fundamentally because $\epsilon(\gamma^{\mathrm{imp}},k,n)$ is not convex on $\gamma^{\mathrm{imp}}$. Thus, we resort to numerical evaluation. We found that $p_{\mathrm{out}}(P_{\mathrm{csi}})$ has the shape of an inverted bell, in part because for small $P_{\mathrm{csi}}$, the mean SNR is low and $\mathbb{E}[\epsilon(\gamma^{\mathrm{imp}},k,n)]\rightarrow 1$; while for large $P_{\mathrm{csi}}$, $p_{_S}\rightarrow 1$. In both cases $p_{\mathrm{out}}\rightarrow 1$ (see \eqref{pout}), with equality from a point onwards and from a point backwards. Based on the inverted bell shape, the minimum is unique, and it is found numerically in Subsection~\ref{imperfectEst}, while analyzing 
its dependence on some system parameters.
\subsection{Relaying scheme}
The equivalent instantaneous SNRs at $D$ and $R$ during the first and second WIT phases with iCSI are
\begin{align}
\setcounter{equation}{54}
\gamma_{_{D1}}^{\mathrm{imp}}&=(\phi_{_{D1}} \min(g_{_{DS}},\lambda_{_S})-\varphi_{_{D1}}) g_{_{SD}}\cdot\nonumber\\
&\cdot
\frac{P_{\mathrm{csi}}u}{P_{\mathrm{csi}}u\!+\!\kappa_{_{SD}}\sigma_{_D}^2(\phi_{_{D1}}\min(g_{_{DS}},\lambda_{_S})-\varphi_{_{D1}}+1)},\label{SNRid1}\\
\gamma_{_{R}}^{\mathrm{imp}}&=(\phi_{_{R}} \min(g_{_{DS}},\lambda_{_S})-\varphi_{_{R}}) g_{_{SR}}\cdot\nonumber\\&
\cdot
\frac{P_{\mathrm{csi}}u}{P_{\mathrm{csi}}u+\kappa_{_{SR}}\sigma_{_R}^2(\phi_{_{R}}\min(g_{_{DS}},\lambda_{_S})-\varphi_{_{R}}+1)},\label{SNRir}
\end{align}
\begin{align}
\gamma_{_{D2}}^{\mathrm{imp}}&=(\phi_{_{D2}} \min(g_{_{DR}},\lambda_{_R})-\varphi_{_{D2}})g_{_{RD}}\cdot\nonumber\\
&\cdot
\frac{P_{\mathrm{csi}}u}{P_{\mathrm{csi}}u\!+\!\kappa_{_{RD}}\sigma_{_D}^2(\phi_{_{D2}}\min(g_{_{DR}},\lambda_{_R})\!-\!\varphi_{_{D2}}+1)},\label{SNRid2}
\end{align}
where $\gamma_{_{D1}}^{\mathrm{imp}}$ is obtained by using \eqref{SNRdc} with $n=n_1$, $\gamma_{_{R}}^{\mathrm{imp}}$ comes from \eqref{SNRi} with $i=S$, $j=R$ and using \eqref{P_s}, and $\gamma_{_{D2}}^{\mathrm{imp}}$ comes from \eqref{SNRi} with $i=R$, $j=D$ and using \eqref{P_r}. Now, the information error probability for each link can be easily calculated as in the pCSI case, but with $\gamma_{_{D1}}=\gamma_{_{D1}}^{\mathrm{imp}}$, $\gamma_{_{R}}=\gamma_{_{R}}^{\mathrm{imp}}$ and $\gamma_{_{D2}}=\gamma_{_{D2}}^{\mathrm{imp}}$.
\section{Numerical Results}\label{results}
In this section, we present numerical results to investigate the performance of the system at finite blocklength with pCSI and iCSI at the receivers.  Let $T_c=\SI{2}{\micro\second}$, thus $\sigma_{_D}^2=-110$dBm is a valid assumption if a bandwidth around 1MHz is assumed. We consider scenarios with stringent error probability and delay requirements, which are expected to be typical of URC-S services in future 5G systems. Therefore, being $\varepsilon_{_0}$ the target error probability and $\delta_{_0}$ the maximum allowable delay,  $p_{\mathrm{out}}\le\varepsilon_{_0}$ and $\delta\le\delta_0$ must be satisfied. Let $\delta_0=8$ms$\rightarrow4000$ channel uses and $\varepsilon_{_0}\in\{10^{-3},10^{-4}\}$.

Results are obtained by setting $\alpha=2$, $d_{_{SR}}=d_{_{SD}}=d_{_{DR}}=d$ and $\kappa=10^3$, what is equivalent to 30 dB average signal power attenuation at a reference distance of 1 meter. 
Sensor nodes, $S$ and $R$, are ultra-low consumption devices\footnote{Notice that sensors in that order of consumption in active mode already exist. Some examples can be found in \cite{Hestnes.2016}.} with $P_c=-30$dBm, and, following the state-of-the-art in circuit design, we consider $\eta=0.5$~\cite{Lu.2015}. Moreover, $P_{_D}=50$dBm and, unless stated otherwise, $k=256$ bits. The utilized values of the system parameters are summarized in Table~\ref{table:SystemParameters}.

\begin{table}[!h]
	\centering
	\caption{System parameters}
	\label{table:SystemParameters}
	\begin{tabular}{lc|lc}
		\hline		
		\textbf{Parameter}	& \textbf{Value}	 & \textbf{Parameter}	& \textbf{Value}\\
		\hline
		$T_c$				& $2\mu$s & $\kappa$ & $10^3$		\\
		$\sigma_{_D}^2$		& $-110$dBm	& $P_c$ & $-30$dBm \\
		$\delta_{_{0}}$			& 8ms	& $\eta$ & $0.5$\\
		$\varepsilon_{_{0}}$	& $\{10^{-3},10^{-4}\}$ & $P_{_D}$ & $50$dBm	\\
		$\alpha$			& $2$	& $k$& $256$\\
		$d$				& $\{10,40\}$m & $u$& $1$\\
		\hline
	\end{tabular}
\end{table}

\subsection{Ideal System}\label{S12}
Herein we analyze the performance of an ideal system, where all the harvested energy is used for information transmission and the receiver has perfect knowledge of the channel while $P_{\mathrm{csi}}=P_c=0$ and $p_{_S}=p_{_R}=0$. Notice that the outage probability equals the error probability for such systems. The performance under such assumptions, which are the most common in the literature, offers an upper-bound for the performance of practical systems. Numerical results are obtained with $d=40$m.

Fig.~\ref{Fig4a} shows the minimum delay, $\delta^*=\min\limits_{p_{_{\mathrm{out}}}\le\varepsilon_{_{0}}}\delta$, required to deliver messages of $k=256$ bits while meeting the reliability constraints given by $\varepsilon_{_{0}}$. DC, and cooperative schemes with MRC ($n_1=n_2=n/2$) and SC ($n_1=n_2=n/2$, $n_1=0.2n$ and $n_2=0.8n$, $n_1=0.8n$ and $n_2=0.2n$), are compared. Here we assume infinite batteries while some analyzes assuming finite batteries are discussed later. Notice that the DC scheme is unable to reach a reliability around  $99.99\%$ ($\varepsilon_{_0}=10^{-4}$) for a maximum allowable delay of 4000 channel uses. However, cooperation through $R$  solves that problem reaching the desirable results. This is because the attained spatial diversity can improve either communication reliability, or can reduce the system delay for a target outage constraint as is the case here.
Among the cooperative schemes, the SC setup with $n_1<n_2$ presents the worst performance since the chances of communication errors increase in both $S\rightarrow D$ and $S\rightarrow R$ links simultaneously. While the opposite occurs for $n_1>n_2$, although slightly.
We can see that the optimum value for $n$ is small, $n^*\le 600$ channel uses, for the cooperative scenarios. On the other hand, the performance in terms of the time sharing parameter, $\beta$, is shown in Fig.~\ref{Fig4b}. When $n$ increases, the required value for $v$ tends to decrease, and therefore $\beta$ decreases too. As the reliability requirements are more stringent, the WET phase must be larger, which is even more accentuated for the DC scheme. Notice that we use the same line to denote the MRC and SC techniques with $n_1=n_2=n/2$ since their performance is almost identical, which is expected from the comments made in Section~\ref{mrc_sc}.  
\begin{figure}[t!] 
	\centering
	\begin{subfigure}[b]{0.45\textwidth}
		\centering
		\includegraphics[width=0.95\textwidth]{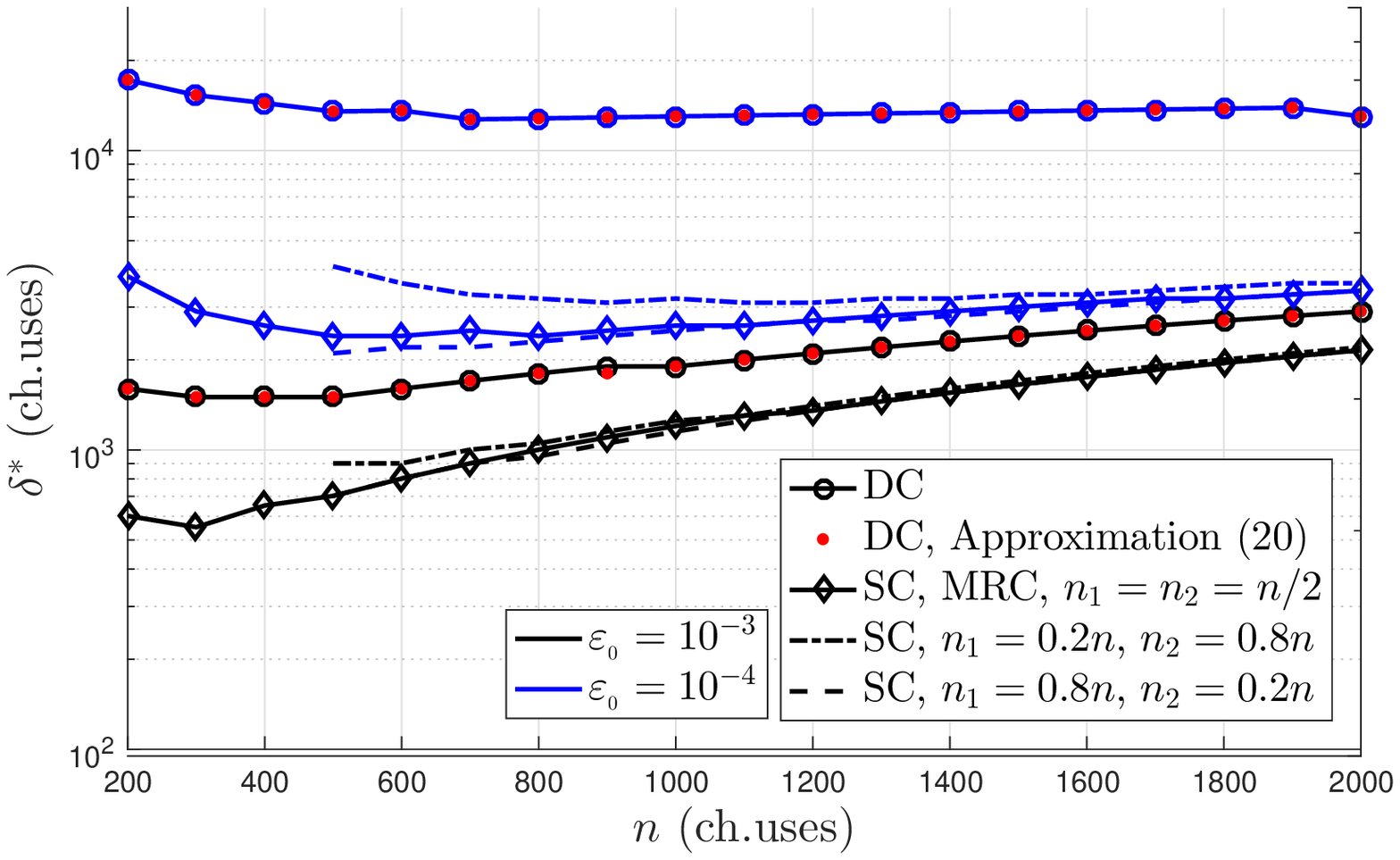}
		\caption{\label{Fig4a}}
	\end{subfigure}    
    \par\bigskip
   \begin{subfigure}[b]{0.45\textwidth}
		\centering
		\includegraphics[width=0.95\textwidth]{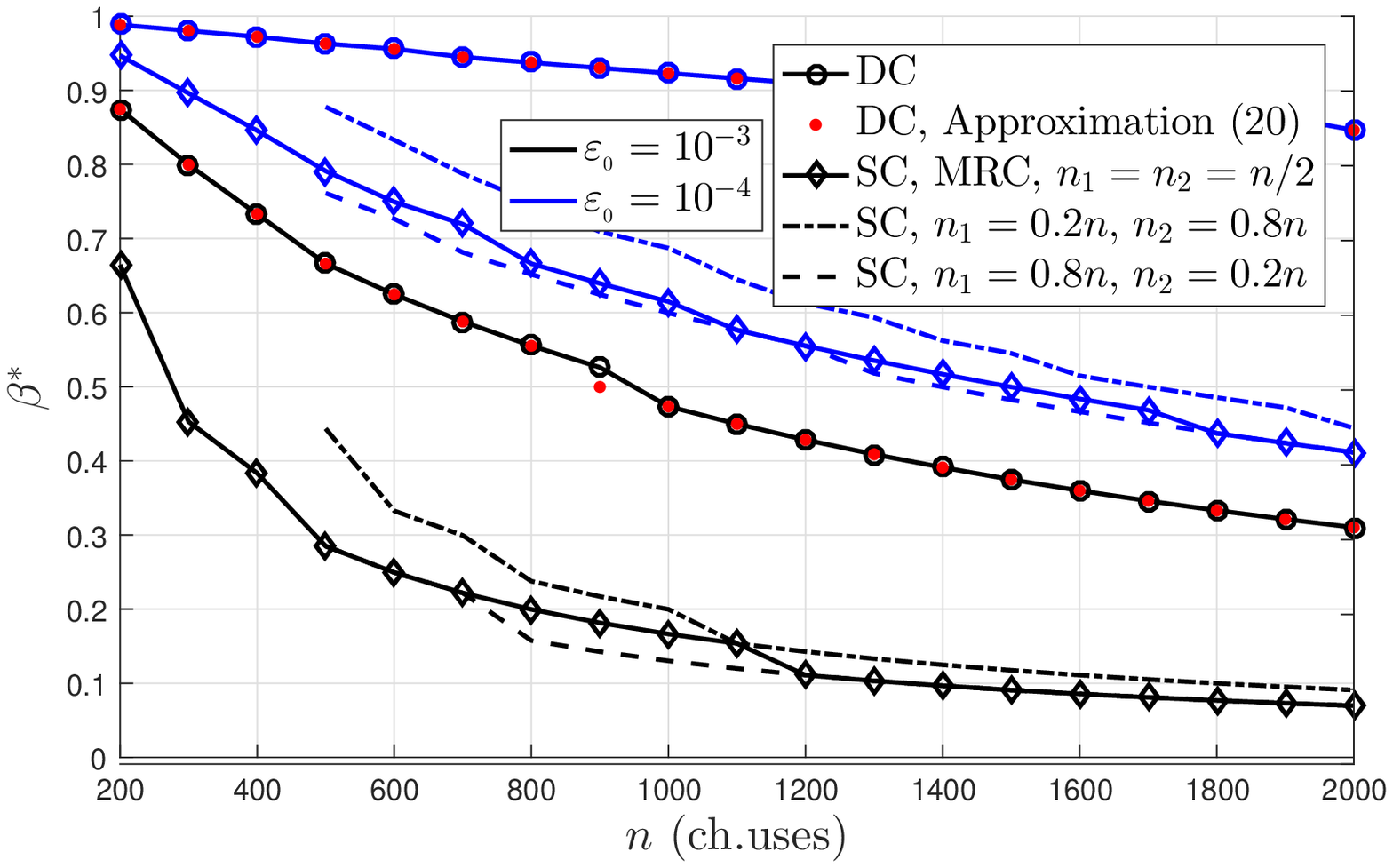}
		\caption{\label{Fig4b}}
	\end{subfigure}
    \vspace{-1mm}		
	\caption{\small (\subref{Fig4a}) $\delta^*$,  (\subref{Fig4b}) $\beta^*$, as a function of $n$ with $B_{_{\mathrm{max}}}=\infty$.}	
	\label{Fig4}	
	\vspace{-2mm}
\end{figure}

The minimum reachable outage probability, $p_{\mathrm{out}}^*=\min\limits_{\delta}p_{\mathrm{out}}$, as a function of the system delay is shown in Fig.~\ref{Fig5a}. According to the figure, it is possible to reach a reliability of $99.995\%$ ($\varepsilon_{_0}=5\times10^{-5}$) with a delay of 8ms (4000 channel uses) when $R$ cooperates with $S$. Without its assistance, the required delay would be higher than 10ms (5000 channel uses) and it is not shown in the figure.
As shown in Fig.~\ref{Fig5b}, while reducing the message length, the minimum delay, required for a given reliability, decreases. Thus, it is possible to achieve a reliability around $99.99\%$ without the relay assistance as long as the messages have no more than $k=80$ bits. For a given reliability and based on \eqref{Q}, when $k$ increases, $n$ should also increase to avoid increasing the rate $r$, and consequently $\gamma$ and $C(\gamma)$ tend to increase. In addition, increasing $v$ decelerates the capacity reduction and $n$ is not required to be so high, thus there is a trade-off between increasing $n$ and $v$. Whatever the case, increasing $k$ for a given $\varepsilon_{_0}$ renders an unavoidable increase in the total number of required channel uses, hence larger delay. 
On the other hand, Figs.~\ref{Fig4} and \ref{Fig5} corroborate the accuracy of expression (20), claimed when discussing Fig.~\ref{Fig2}.
\begin{figure}[t!] 
	\centering
	\begin{subfigure}[b]{0.45\textwidth}
		\centering
		\includegraphics[width=0.95\textwidth]{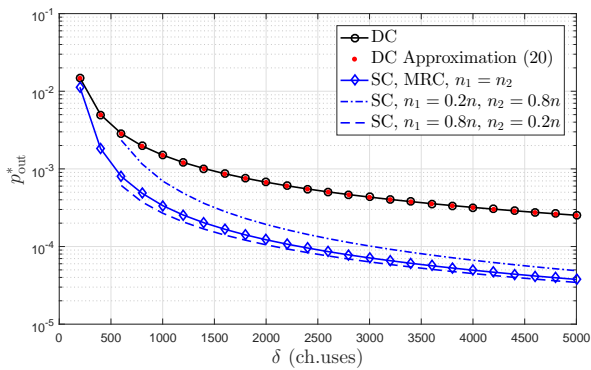}		
		\caption{\label{Fig5a}}
	\end{subfigure}    
	\par\bigskip
	\begin{subfigure}[b]{0.45\textwidth}
		\centering
		\includegraphics[width=0.95\textwidth]{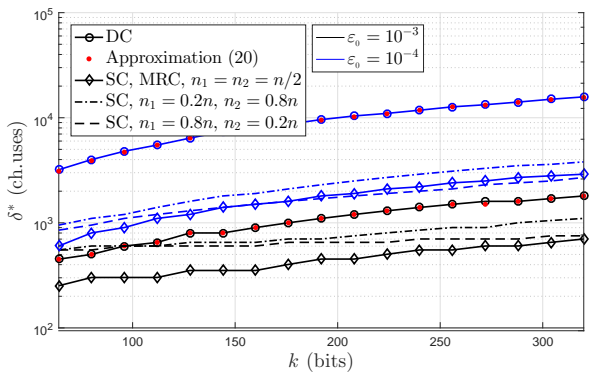}
		\caption{\label{Fig5b}}
	\end{subfigure}		
	\caption{\small (\subref{Fig5a}) $p_{\mathrm{out}}^*$, (\subref{Fig5b}) $\delta^*$,  as a function of $\delta$,  $k$, with  $B_{_{\mathrm{max}}}\!=\!\infty$.}	
	\label{Fig5}	
	\vspace{-2mm}
\end{figure} 
\subsection{Non-Ideal channel estimation}
Herein, we analyze a more realistic case, where part of the harvested energy is required by the devices (namely $S$ and $R$) to feed their circuits and to transmit the pilot signals in order to acquire CSI at the receivers. The imperfection of the acquired CSI is also taken into account.
\begin{figure}[t!] 
	\centering
	\begin{subfigure}[b]{0.45\textwidth}
		\centering
		\includegraphics[width=0.95\textwidth]{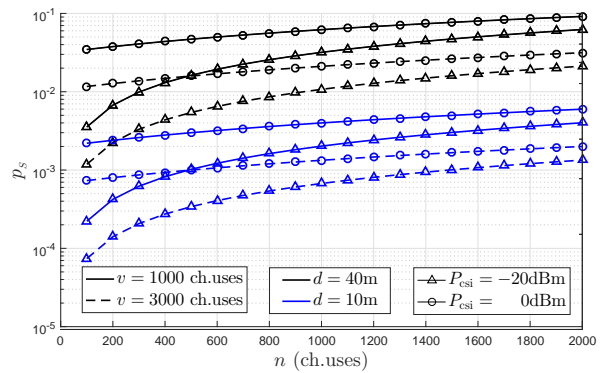}
		\caption{\label{Fig6a}}
	\end{subfigure}    
	\par\bigskip
	\begin{subfigure}[b]{0.45\textwidth}
		\centering
		\includegraphics[width=0.95\textwidth]{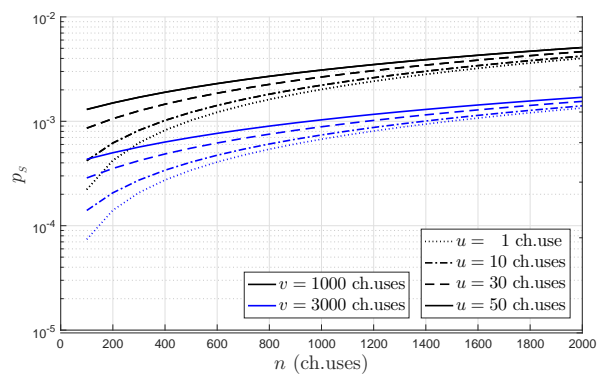}
		\caption{\label{Fig6b}}
	\end{subfigure}		
	\caption{\small $p_{_S}$ as a function of $n$ with $B_{_{\mathrm{max}}}=\infty$. (\subref{Fig6a}) $d\in\{10,40\}$m, $P_{\mathrm{csi}}\in\{0,-20\}$dBm and $u=1$ channel use, (\subref{Fig6b}) $d=10$m and $P_{\mathrm{csi}}=-20$dBm.	}
	\label{Fig6}		
	\vspace{-2mm}
\end{figure}

\subsubsection{Outages due to insufficient energy harvested}\label{S21}
Fig.~\ref{Fig6} shows the probability that the harvested energy at $S$ is insufficient simultaneously for channel estimation and for satisfying other consumption requirements at $S$. Notice that the curves appear as functions of $n$, where $n=n_1$ for the cooperative case, and $v\in\{1000,3000\}$ channel uses. In Fig.~\ref{Fig6a}, different power levels for channel estimation were tested, $P_{\mathrm{csi}}\!=\!\{0,\ \!-\!20\}$dBm, with $u\!=\!1$ channel use. Thus, $P_{\mathrm{csi}}T_c$ is equivalent to the energy used for CSI estimation. Notice that when $n$ increases, $p_{_S}$ also increases, which is expected from \eqref{ps}, since the sensor remains active longer and its circuits require more energy. Moreover, a larger value of $v$ leads to more energy being harvested, decreasing $p_{_S}$. Also, we can notice that it is impossible to reach a high reliability when $d=40$m and consequently from now on we use $d=10$m. The performance increases by decreasing $P_{\mathrm{csi}}$, specially for a small $n$. Although it is not shown in the figure, the performance improves much slower by decreasing $P_{\mathrm{csi}}$ below $-20$dBm. Thus, it seems convenient to always use an energy no greater than $uT_cP_{\mathrm{csi}}\big|_{u=1,P_{\mathrm{csi}}=-20\mathrm{dBm}}$, although we have to keep in mind that when $P_{\mathrm{csi}}$ decreases, the communication errors increase due to a greater imperfection on the channel estimates, which is analyzed later. On the other hand, Fig.~\ref{Fig6b} shows $p_{_S}$ when $P_{\mathrm{csi}}=-20$dBm and $u=\{1,\ 10,\ 30,\ 50\}$ channel uses. When $u$ increases, the required energy for CSI estimate increases thus, the probability that the harvested energy is insufficient for attending those requirements also increases. This dependency starts decreasing when $n$ increases since the influence of $u$ on the delay and the active sensor phase reduces. We argue that, from the standpoint of $p_{_S}$, the more suitable strategy is to set $u=1$ channel use and properly choose $P_{\mathrm{csi}}$ in order to minimize the impact on $P_c$, which is confirmed when analyzing the communication error expressions.
\begin{figure}[t!] 
	\centering
	\begin{subfigure}[b]{0.45\textwidth}
		\centering
		\includegraphics[width=0.95\textwidth]{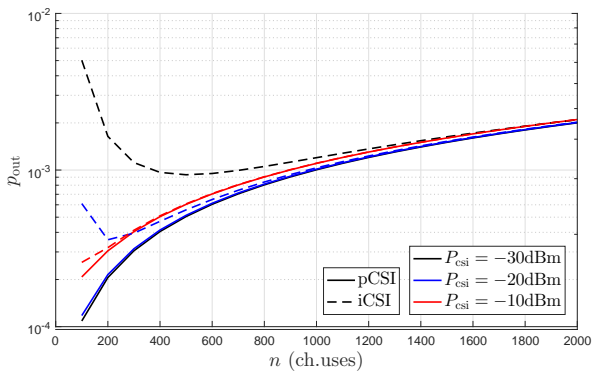}
		\caption{\label{Fig7a}}
	\end{subfigure}    
	\par\bigskip
	\begin{subfigure}[b]{0.45\textwidth}
		\centering
		\includegraphics[width=0.95\textwidth]{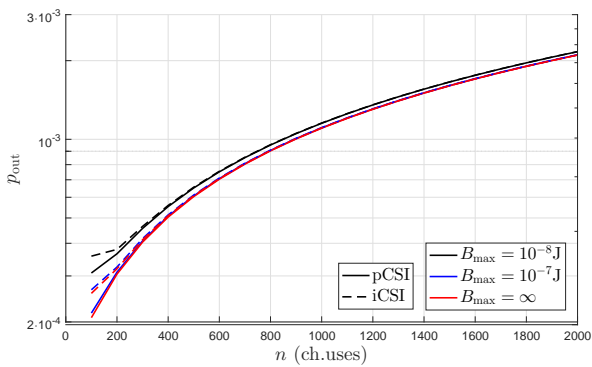}
		\caption{\label{Fig7b}}
	\end{subfigure}
	\caption{\small $p_{\mathrm{out}}$ as a function of $n$ with $u=1$ and $v=2000$ channel uses. (\subref{Fig7a}) $P_{\mathrm{csi}}\in\{-30,-20,-10\}$dBm and $B_{_{\mathrm{max}}}=\infty$,  (\subref{Fig7b}) $P_{\mathrm{csi}}=-10$dBm and $B_{_{\mathrm{max}}}\in\{10^{-8},10^{-7},\infty\}$J.}	
	\label{Fig7}		
	\vspace{-2mm}
\end{figure}

\subsubsection{On the imperfect estimates and finite battery}\label{imperfectEst}
Fig.~\ref{Fig7} shows the DC system performance in terms of outage probability as a function of $n$, and considering $u=1,\ v=2000$ channel uses. In Fig.~\ref{Fig7a}, we set $B_{_{\mathrm{max}}}=\infty$ and evaluate the performance for $P_{\mathrm{csi}}\in\{-30,-20,-10\}$dBm when the estimates are considered perfect and imperfect. The performance gap when assuming pCSI and iCSI is very perceptible for low pilot power, which means that the used energy is still insufficient and the number of pilot symbols and/or its transmit power must be increased. Notice that a proper energy value is around $-10\,$dBm$\cdot T_c$, which is larger than the one we had recommended in Section~\ref{S21} since now the channel estimation affects the error probability when there is communication, and a relatively larger value is required. Consequently, the results shown in Fig.~\ref{Fig7b} are with $P_{\mathrm{csi}}=-10$dBm, but now considering the finite battery case. Notice that, when $B_{_{\mathrm{max}}}=10^{-7}$J the system performance approximates the case when $B_{_{\mathrm{max}}}=\infty$
(blue and red curves are almost coincident). The general performance gets worse when $B_{_{\mathrm{max}}}=10^{-8}$J, although practical batteries have storage capacities much greater than the values under consideration here, which means that models assuming infinite battery are adequate in this context.
\begin{figure}[t!] 
	\centering
	\begin{subfigure}[b]{0.45\textwidth}
		\centering
		\includegraphics[width=0.95\textwidth]{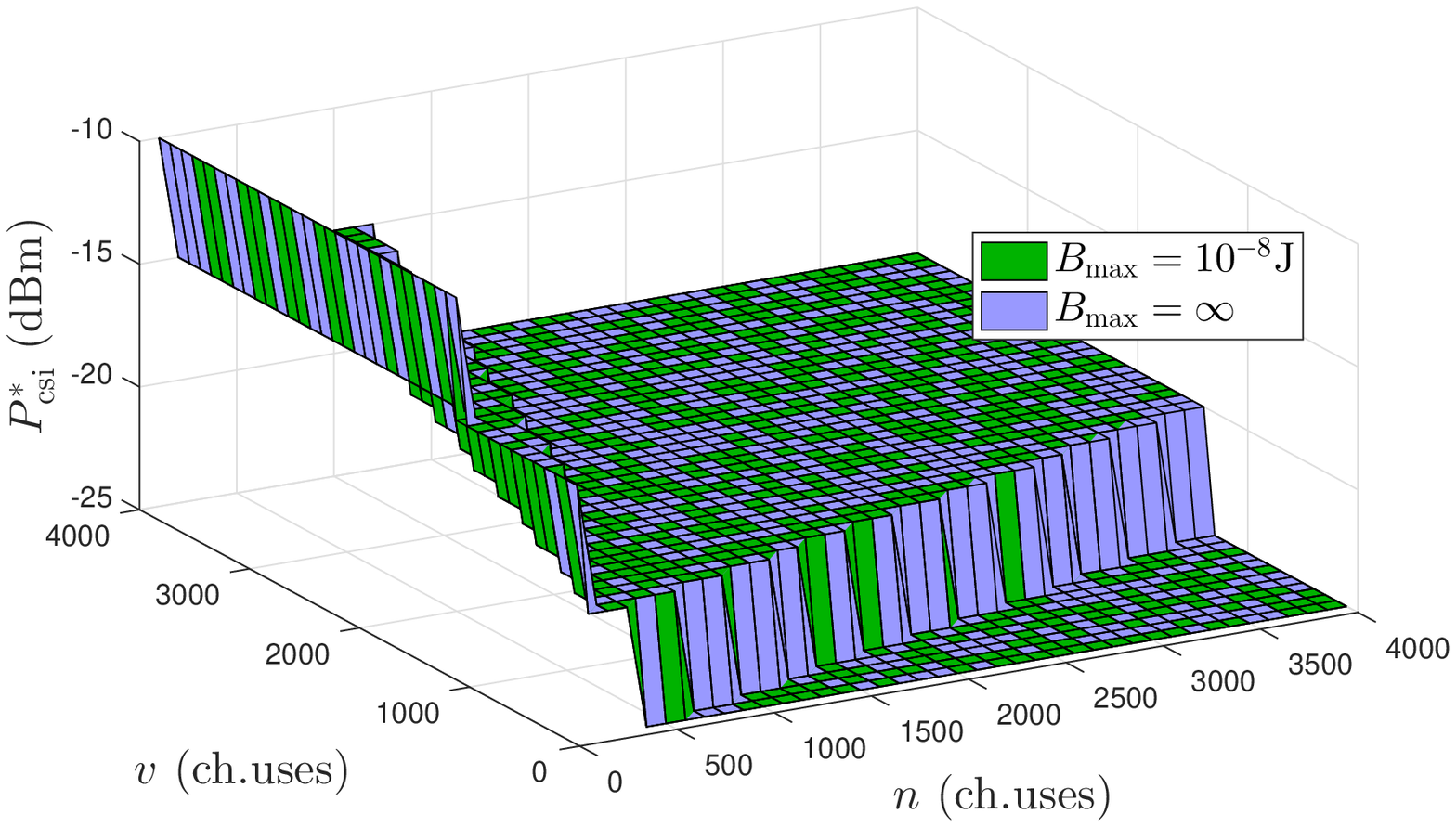}
		\caption{\label{Fig8a}}
	\end{subfigure}    
	\par\bigskip
	\begin{subfigure}[b]{0.45\textwidth}
		\centering
		\includegraphics[width=0.95\textwidth]{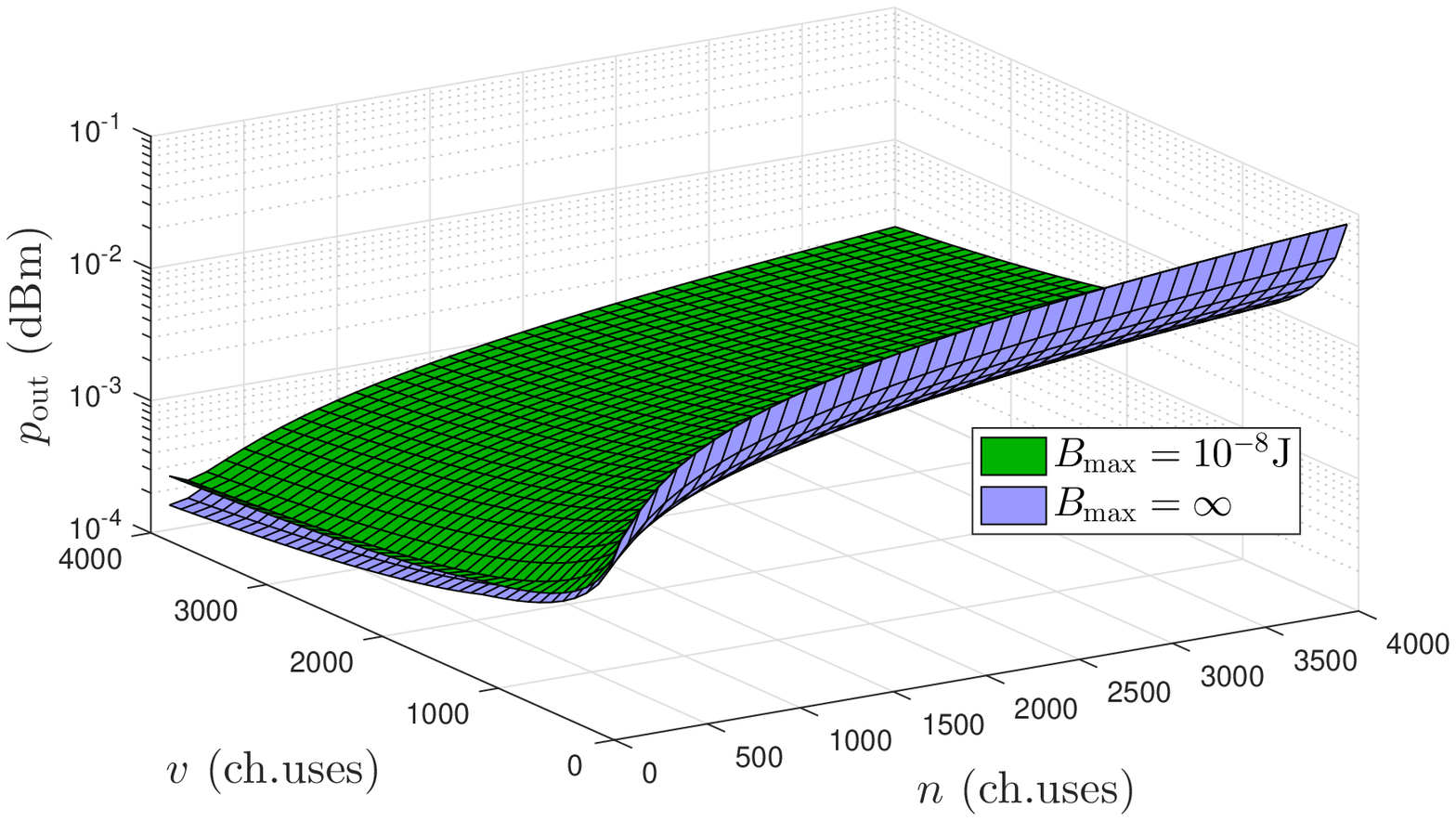}
		\caption{\label{Fig8b}}
	\end{subfigure}
	\caption{\small (\subref{Fig8a}) $P_{\mathrm{csi}}^*$ and (\subref{Fig8b}) $p_{\mathrm{out}}$, as functions of $n$ and $v$ with $u=1$ channel use.	}
	\label{Fig8}
	\vspace{-2mm}
\end{figure} 

Fig.~\ref{Fig8a} shows the optimum power for pilot transmission, $P_{\mathrm{csi}}^*=\arg\min\limits_{P_{\mathrm{csi}}}p_{\mathrm{out}}$, when $u=1$ channel use, and  Fig.~\ref{Fig8b} shows the outage probability reached for those values.	 
Interestingly, $P_{\mathrm{csi}}^*$ depends on the pair $(n,v)$. Configurations with a large $v$ also require a large $P_{\mathrm{csi}}$, and an even larger pilot transmission power is needed when $n$ is small. On the other hand, for large $n$ a small $P_{\mathrm{csi}}$ is adequate, which is evinced when $v$ is small. As shown in Section~\ref{S12}, the optimal performance occurs for relatively small and large $n$ and $v$, respectively, and from now on we use $-15$dBm as the pilot transmit power, which seems appropriate according to Fig.~\ref{Fig8}. Once again, a small performance gap can be noticed when comparing configurations with finite and infinite battery, reinforcing the convenience of the analysis with infinite battery assumption for these scenarios.

\subsubsection{On the performance of the cooperative scheme}\label{prev}
In this subsection we analyze scenarios where devices are assumed to be equipped with infinite battery. Fig.~\ref{Fig9a} shows the minimum delay required to deliver the messages while reaching the reliability constraints. Notice that
a reliability around $99.99\%$ ($\varepsilon_{_0}=10^{-4}$), is impossible to be reached without the relay assistance within the allowable delay (10ms). Among the different combination configurations at $D$, the SC scheme\footnote{Notice that $n\ge500$ channel uses since $n_1=0.2\cdot500=100$ channel uses, which is the lower bound to an accurate approximation when using \eqref{Q} as demonstrated in \cite[eq.(59)]{Yang.2014} for quasi-static fading channels and sufficiently large values of $n\ge100$.} with $n_1<n_2$ performs better for a given $n$. Obviously, the smaller $n_1$ is, the longer is the time that $S$ remains inactive, reducing the energy consumption of its circuits. Differently from the ideal case discussed in Section~\ref{S12}, $n^*<200$ channel uses and it is not shown in the figure. SC and MRC with $n_1=n_2$ perform similarly, while SC with $n_1>n_2$ forces a higher energy consumption of the $S$ circuits, remaining less energy resources for pilot and information transmission, although its performance still overcomes the DC case. Increasing $P_c$, the performance gap among the SC schemes with $n_1>n_2$, $n_1=n_2$ and $n_1<n_2$ also increases, since the impact of $n$ would be greater. As shown in Fig.~\ref{Fig9b}, the larger $n$ is, the smaller is the required fraction of time for WET, although a slight decrease is noticeable\footnote{Notice that in the ideal system scenario discussed in Section~\ref{S12}, that decrease is much pronounced.} due to the non-trivial value of $P_c$. In order to achieve a high reliability, the larger energy expenditure is required, thus $\beta$ gets closer to unity. Finally, the performance gap among the cooperative schemes with different blocklength in the broadcast and cooperation phase becomes more relevant when reliability constraints are more stringent.
\begin{figure}[t!] 
	\centering
	\begin{subfigure}[b]{0.45\textwidth}
		\centering
		\includegraphics[width=0.95\textwidth]{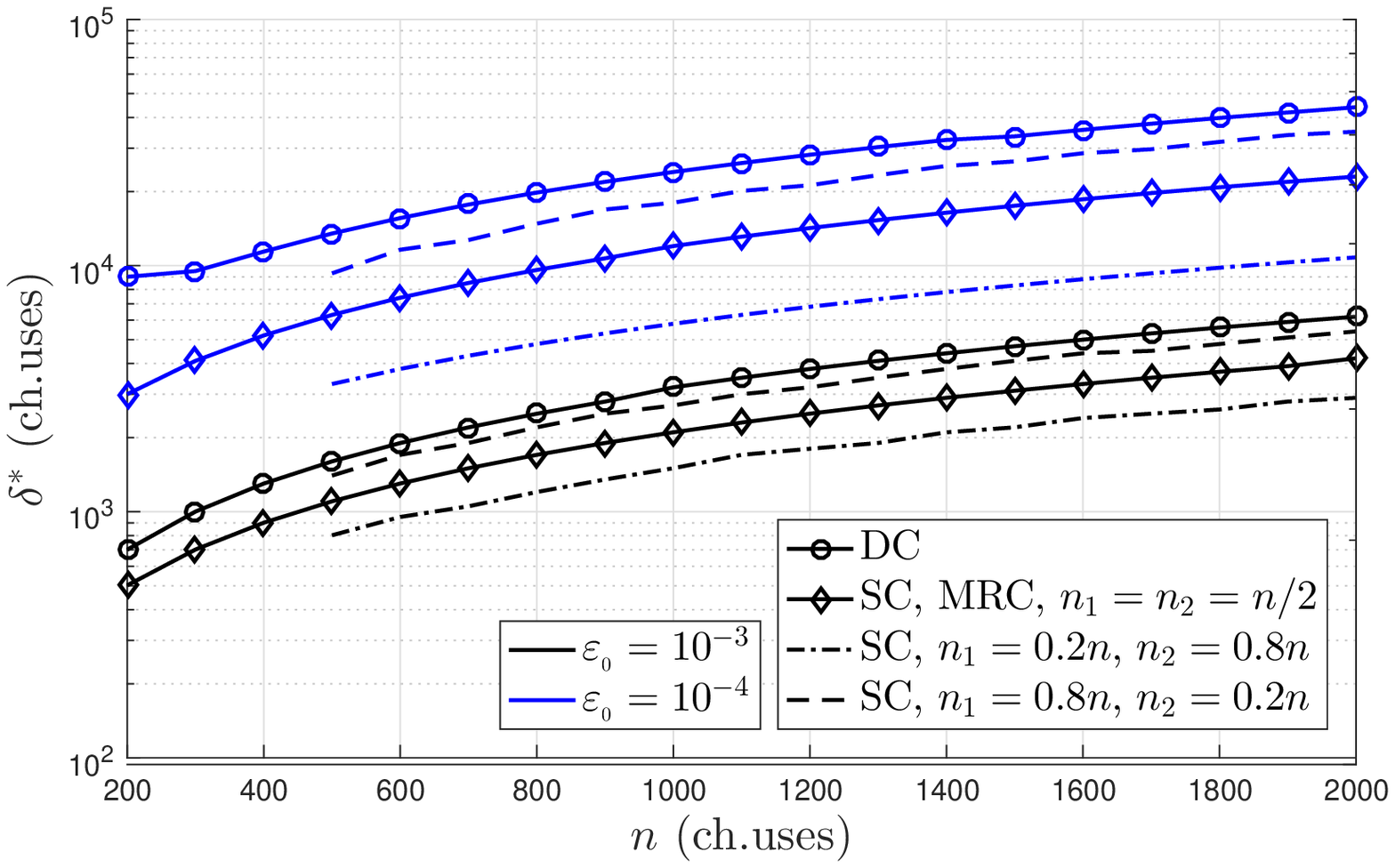}
		\caption{\label{Fig9a}}
	\end{subfigure}    
	\par\bigskip
	\begin{subfigure}[b]{0.45\textwidth}
		\centering
		\includegraphics[width=0.95\textwidth]{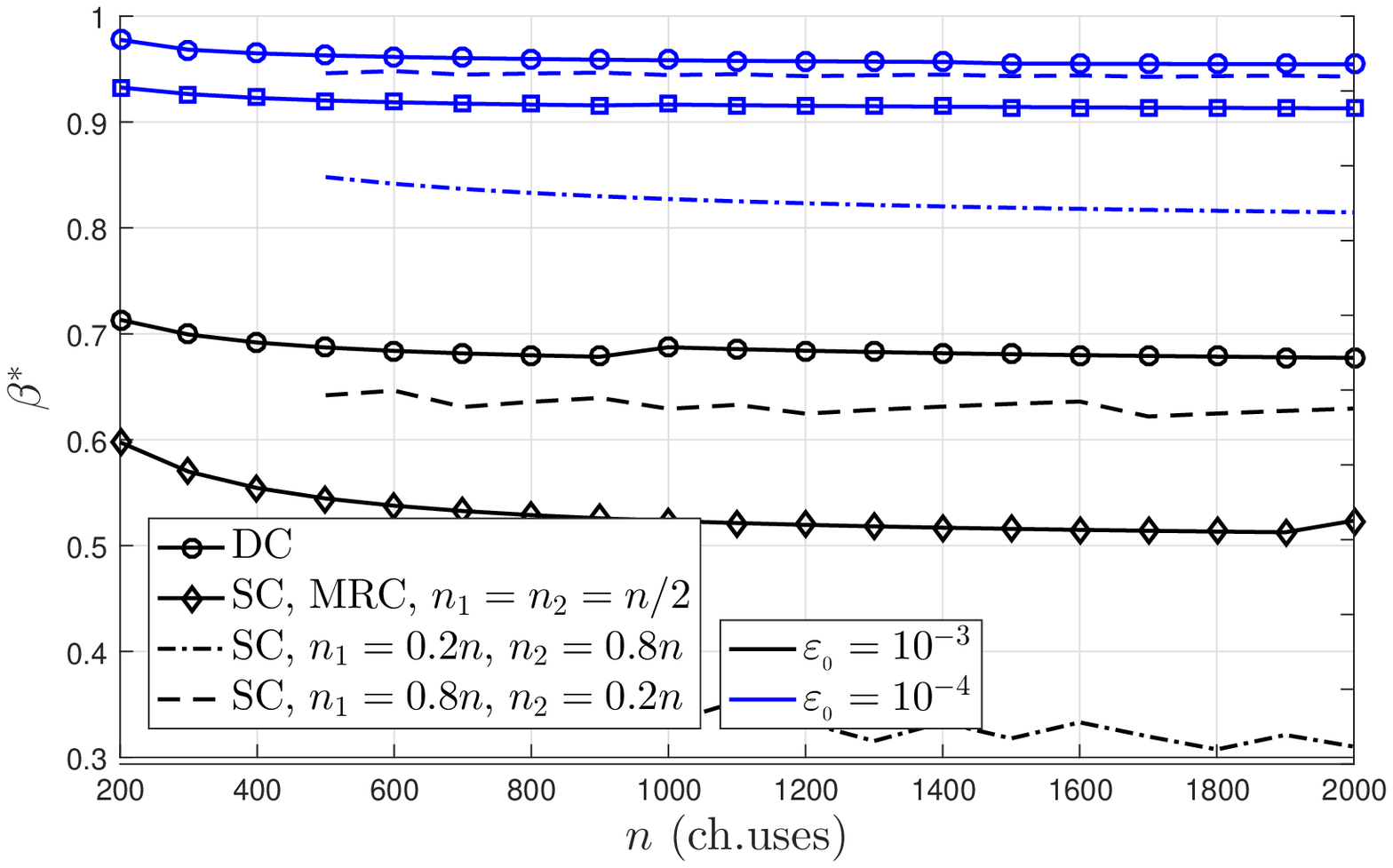}
		\caption{\label{Fig9b}}
	\end{subfigure}
	\caption{\small (\subref{Fig9a}) $\delta^*$,  (\subref{Fig9b}) $\beta^*$, as a function of $n$ with $B_{_{\mathrm{max}}}=\infty$.	}
	\label{Fig9}		
	\vspace{-2mm}
\end{figure}

The minimum reachable outage probability is shown in Fig.~\ref{Fig10a}, while the relative frequency of the outages due to the energy insufficiency at $S$ as a function of the message length is shown in Fig.~\ref{Fig10b}. We set $n=500$ channel uses. According to Fig.~\ref{Fig10a}, a delay of 10ms (5000 channel uses) when using the DC or cooperative schemes with $n_1\ge n_2$, is insufficient to reach a reliability around $99.99\%$. For those schemes, the required delay would have to be much higher to fulfill such requirement, specially for the DC case. With SC and $n_1=100$, $n_2=400$ channel uses, that requirement can be attended with a delay around $6.4$ms ($\sim3200$ channel uses). The convenience of choosing $n_1<n_2$ was already discussed above. On the other hand, a main conclusion coming from Fig.~\ref{Fig10b} is that increasing $k$ would not significantly degrade the performance for the cooperative schemes since the main cause of outage is due to energy insufficiency, meanwhile for the DC scheme, increasing $k$ up to $300$ bits could decrease the performance since the communication errors become almost $18\%$ of the outage events. 
\begin{figure}[t!] 
	\centering
	\begin{subfigure}[b]{0.45\textwidth}
		\centering
		\includegraphics[width=0.95\textwidth]{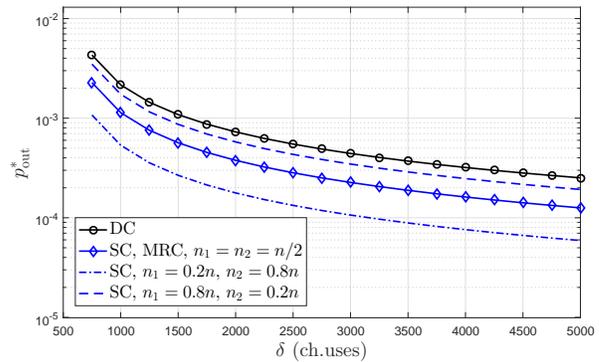}
		\caption{\label{Fig10a}}
	\end{subfigure}    
	\par\bigskip
	\begin{subfigure}[b]{0.45\textwidth}
		\centering
		\includegraphics[width=0.95\textwidth]{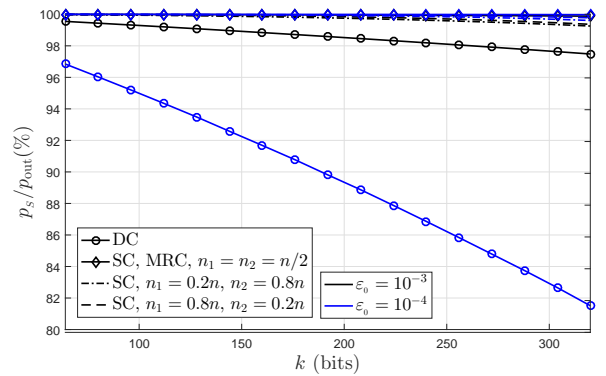}
		\caption{\label{Fig10b}}
	\end{subfigure}
	\vspace{-1mm}		
	\caption{\small (\subref{Fig10a}) Minimum reachable outage probability as a function of the delay, (\subref{Fig10b}) Relative frequency of outages due to energy insufficiency at $S$ as a function of the message length, with $P_{\mathrm{csi}}=-15$dBm, $B_{_{\mathrm{max}}}=\infty$, $u=1$ and $n=500$ channel uses.	}	
	\vspace{-2mm}
\end{figure} 
\section{Conclusion}\label{conclusions}
We evaluated a cooperative wireless-powered communication network at finite blocklength regime and with limited battery capacity. We modeled a realistic system where other power consumption sources beyond data transmission, such as circuit and baseband processing and pilot transmission, are taken into account along with imperfections on channel estimates. We also characterized the error probability for direct communication and cooperative schemes, while attaining closed-form expressions in ideal direct communication systems. 

Our results demonstrate that there is an optimum pilot transmit power for channel estimation, which depends on system parameters such as $n$, $v$, $P_{_D}$ and $k$. We show that, whenever more energy is harvested, the optimum pilot energy must be greater, and even more for small $n$. An important remark is that the energy devoted to the CSI acquisition, which depends on the power and utilized time, has to be taken into account due to the inherent energy and delay constraints of the discussed scenarios. Also, while infinite battery assumption is shown to be permissible for the scenarios under discussion since the energy harvested is very low; considering other power consumption sources beyond data transmission is crucial because they constitute a non-negligible cause of outage.
Moreover, we show that cooperation is important in order to reach URC-S requirements, while also illustrating the convenience of using a small blocklength during the first WIT phase, since $S$ stays active for less time and thus decreases its power consumption.
Finally, a greater flexibility regarding the latency target allows for reaching stringent reliability constraints and vice-versa. 
 
As a future work we intend to analyze multi-relay scenarios for short message communications, while proposing practical schemes. It could be also interesting to incorporate other strategies to these scenarios such as power allocation and/or HARQ.
\appendices 
\section{Proof of Theorem~\ref{Th1}}\label{App_A}
 Let $x=g_{_{DS}}$ and $y=g_{_{SD}}$ for shortening notation. Since the system is considered ideal, we have $p_s=0\rightarrow p_{\mathrm{out}}=\mathbb{E}[\epsilon(\gamma,k,n)]$, and $\gamma=\phi y\min(x,\lambda_{_S})$. We resort to the first order approximation of  $Q(f(\gamma))$, $f(\gamma)=\frac{C(\gamma)-r}{\sqrt{V(\gamma)/n}}$, given by \cite{Makki.2014,Makki.2016}
 \begin{eqnarray}\label{AP}
 Q(f(\gamma))\!\approx\!\Omega(\gamma)\!=\!\left\{\begin{array}{ll}
 1,&  \gamma\le \varrho\\
 \frac{1}{2}\!-\!\frac{\psi}{\sqrt{2\pi}}(\gamma\!-\!\theta)\!,\!&\!\varrho\!<\!\gamma\!<\!\vartheta\!\\
 0,& \gamma\ge \vartheta
 \end{array}
 \right.\!, 
 \end{eqnarray}
 where $\theta=2^{k/n}-1$, $\psi=\sqrt{\frac{n}{2\pi}}(2^{2k/n}-1)^{-\tfrac{1}{2}}$, $\varrho=\theta-\tfrac{1}{\psi}\sqrt{\tfrac{\pi}{2}}$ and $\vartheta=\theta+\tfrac{1}{\psi}\sqrt{\tfrac{\pi}{2}}$.
\subsection{Derivation of \eqref{pout_fin}}
The expression in \eqref{AP} can be reformulated using $\gamma=\phi y\min(x,\lambda_{_S})$ as follows
\begin{align}\label{AP_A1}
&\Omega(\phi y\min(x,\lambda_{_S}))\nonumber\\
&\!=\!\left\{ \begin{array}{ll}
1,&  y\!\le\! \frac{\varrho}{x\phi},\ x<\lambda_{_S}\\
1,&  y\!\l\!e \frac{\varrho}{\lambda_{_S}\phi},\ x\ge\lambda_{_S}\\
\frac{1}{2}\!-\!\frac{\psi}{\sqrt{2\pi}}(\phi xy\!-\!\theta),&   \frac{\varrho}{x\phi}\!<\!y\!<\!\frac{\vartheta}{x\phi},\ x\!<\!\lambda_{_S}\\
\frac{1}{2}\!-\!\frac{\psi}{\sqrt{2\pi}}(\phi\lambda_{_S} y\!-\!\theta),&   \frac{\varrho}{\lambda_{_S}\phi}\!<\!y\!<\!\frac{\vartheta}{\lambda_{_S}\phi},\ x\!\ge\!\lambda_{_S}\\
0,& \mathrm{otherwise}
\end{array}\!
\right.\!.
\end{align}
Now, by substituting  \eqref{AP_A1} into \eqref{pout_fin} since $Q(f(\phi y\min(x,\lambda_{_S})))\approx \Omega(\phi y\min(x,\lambda_{_S}))$, we get

\begin{align}
p_{\mathrm{out}}\approx& \int\limits_{0}^{\lambda_{_S}}\int\limits_{0}^{\frac{\varrho}{x\phi}}e^{-x-y}\mathrm{d}y\mathrm{d}x+\int\limits_{\lambda_{_S}}^{\infty}\int\limits_{0}^{\frac{\varrho}{\lambda_{_S}\phi}}e^{-x-y}\mathrm{d}y\mathrm{d}x+\nonumber\\
&+\int\limits_{0}^{\lambda_{_S}}\int\limits_{\frac{\varrho}{x\phi}}^{\frac{\vartheta}{x\phi}}\bigg(\frac{1}{2}-\frac{\psi}{\sqrt{2\pi}}(\phi xy-\theta)\bigg)e^{-x-y}\mathrm{d}y\mathrm{d}x+\nonumber\\
&+\int\limits_{\lambda_{_S}}^{\infty}\int\limits_{\frac{\varrho}{\lambda_{_S}\phi}}^{\frac{\vartheta}{\lambda_{_S}\phi}}\bigg(\frac{1}{2}-\frac{\psi}{\sqrt{2\pi}}(\phi \lambda_{_S} y-\theta)\bigg)e^{-y-x}\mathrm{d}y\mathrm{d}x\nonumber\\
&\approx I_1+I_2+I_3+I_4,
\label{AP_A12}
\end{align}
where $I_j$ is the $j-$th adding integral with $j=1,2,3,4$, and each of one has to be solved to find a closed-form expression. 
To be able to do so we first attain the following results
\begin{align}
\int\limits_{0}^{c}e^{-x-\frac{d}{x}}&\mathrm{d}x\!\stackrel{(a)}{=}\!\int\limits_{0}^{c}\!\sum_{i=0}^{\infty}\frac{(-1)^ix^i}{i!}e^{-\frac{d}{x}}\mathrm{d}x\!=\!\sum_{i=0}^{\infty}\frac{(-1)^i}{i!}\!\int\limits_{0}^{c}\!x^ie^{-\frac{d}{x}}\mathrm{d}x\nonumber\\
\!\stackrel{(b)}{=}\sum_{i=0}^{\infty}\!&\frac{(-1)^ix^{i+1}\mathrm{Ei}(i+2,\frac{d}{x})}{i!}\bigg|_0^{c}\nonumber\\
=\!\sum_{i=0}^{\infty}\!&\frac{(\!-\!1)^ic^{i\!+\!1}\mathrm{Ei}(i\!+\!2,\frac{d}{c})}{i!}\!-\!\lim\limits_{x\rightarrow 0}\sum_{i=0}^{\infty}\!\frac{(-1)^ix^{i\!+\!1}\mathrm{Ei}(i\!+\!2,\frac{d}{x})}{i!}\nonumber\\
=\!\sum_{i=0}^{\infty}\!&\frac{(-1)^ic^{i+1}\mathrm{Ei}(i+2,\frac{d}{c})}{i!},
\label{Ip1}\\
\int\limits_{0}^{c}\!xe^{-\!x-\!\frac{d}{x}}&\mathrm{d}x\!\stackrel{(a)}{=}\!\int\limits_{0}^{c}\!\sum_{i=0}^{\infty}\!\frac{(\!-\!1)^ix^{i\!+\!1}}{i!}e^{\!-\!\frac{d}{x}}\mathrm{d}x\!=\!\sum_{i=0}^{\infty}\!\frac{(\!-\!1)^i}{i!}\!\int\limits_{0}^{c}\!x^{i\!+\!1}e^{\!-\!\frac{d}{x}}\mathrm{d}x\nonumber\\
\!\stackrel{(b)}{=}\sum_{i=0}^{\infty}\!&\frac{(-1)^ix^{i+2}\mathrm{Ei}(i+3,\frac{d}{x})}{i!}\bigg|_0^{c}\nonumber\\
=\!\sum_{i=0}^{\infty}\!&\frac{(-1)^ic^{i+2}\mathrm{Ei}(i\!+\!3,\frac{d}{c})}{i!}\!-\!\lim\limits_{x\rightarrow 0}\!\sum_{i=0}^{\infty}\!\frac{(-1)^ix^{i\!+\!2}\mathrm{Ei}(i\!+\!3,\frac{d}{x})}{i!}\nonumber\\
=\!\sum_{i=0}^{\infty}\!&\frac{(-1)^ic^{i+2}\mathrm{Ei}(i+3,\frac{d}{c})}{i!},
\label{Ip2}
\end{align}
where $(a)$ comes from the Taylor series expansion of $e^{-x}$ and $(b)$ comes from the definition of the exponential integral  \cite{Jeffrey.2007}. 
Now we proceed solving each $I_j$ integral as shown in \eqref{I1}, \eqref{I2}, \eqref{I3}, \eqref{I4}. The last two are on the top of the next page and the final expressions for $I_1$ and $I_2$ are obtained from using \eqref{Ip1} and \eqref{Ip2} as
\begin{align}
I_1&\!=\!\!\!\int\limits_{0}^{\lambda_{_S}}\!\int\limits_{0}^{\frac{\varrho}{x\phi}}\!e^{\!-\!x\!-\!y}\mathrm{d}y\mathrm{d}x\!=\!\!\!\int\limits_{0}^{\lambda_{_S}}\!\!(e^{\!-\!x}\!-\!e^{\!-\!x\!-\!\frac{\varrho}{x\phi}})\mathrm{d}x\!=\!1\!-\!e^{\!-\!\lambda_{_S}}\!\!-\!\!\int\limits_{0}^{\lambda_{_S}}\!\!e^{\!-\!x\!-\!\frac{\varrho}{x\phi}}\mathrm{d}x\nonumber\label{I1}\\
&=1-e^{-\lambda_{_S}}-\sum_{i=0}^{\infty}\frac{(-1)^i\lambda_{_S}^{i+1}\mathrm{Ei}(i+2,\frac{\varrho}{\lambda_{_S}\phi})}{i!},\\
I_2&=\int\limits_{\lambda_{_S}}^{\infty}\int\limits_{0}^{\frac{\varrho}{\lambda_{_S}\phi}}e^{-x-y}\mathrm{d}y\mathrm{d}x=\int\limits_{\lambda_{_S}}^{\infty}(1-e^{-\frac{\varrho}{\lambda_{_S}\phi}})e^{-x}\mathrm{d}x\nonumber\\
&=(1-e^{-\frac{\varrho}{\lambda_{_S}\phi}})e^{-\lambda_{_S}}=e^{-\lambda_{_S}}-e^{-\lambda_{_S}-\frac{\varrho}{\lambda_{_S}\phi}}.\label{I2}
\end{align}
\begin{figure*}[!t]
\begin{align}
I_3&=\int\limits_{0}^{\lambda_{_S}}\int\limits_{\frac{\varrho}{x\phi}}^{\frac{\vartheta}{x\phi}}\bigg(\frac{1}{2}-\frac{\psi}{\sqrt{2\pi}}(\phi xy-\theta)\bigg)e^{-x-y}\mathrm{d}y\mathrm{d}x=\!\bigg(\!\frac{1}{2}\!+\!\frac{\psi\theta}{\sqrt{2\pi}}\!\bigg)\!\int\limits_{0}^{\lambda_{_S}}\!\int\limits_{\frac{\varrho}{x\phi}}^{\frac{\vartheta}{x\phi}}\!e^{-x-y}\mathrm{d}y\mathrm{d}x\!-\!\frac{\psi\phi}{\sqrt{2\pi}}\!\int\limits_{0}^{\lambda_{_S}}\!\int\limits_{\frac{\varrho}{x\phi}}^{\frac{\vartheta}{x\phi}}\!xye^{-x-y}\mathrm{d}y\mathrm{d}x\nonumber\\
&\!=\!\bigg(\!\frac{1}{2}\!+\!\frac{\psi\theta}{\sqrt{2\pi}}\!\bigg)\!\int\limits_0^{\lambda_{_S}}\!\big(e^{\!-\!x\!-\!\frac{\varrho}{\phi x}}\!-\!e^{\!-\!x\!-\!\frac{\vartheta}{\phi x}}\big)\mathrm{d}x\!-\!\frac{\psi\phi}{\sqrt{2\pi}}\!\int\limits_0^{\lambda_{_S}}\!xe^{\!-\!x\!-\!\frac{\varrho}{\phi x}}\mathrm{d}x\!-\!\frac{\psi\varrho}{\sqrt{2\pi}}\!\int\limits_{0}^{\lambda_{_S}}\!e^{\!-\!x\!-\!\frac{\varrho}{\phi x}}\mathrm{d}x\!+\!\frac{\psi\phi}{\sqrt{2\pi}}\!\int\limits_0^{\lambda_{_S}}\!xe^{\!-\!x\!-\!\frac{\vartheta}{\phi x}}\mathrm{d}x\!+\!\frac{\psi\vartheta}{\sqrt{2\pi}}\!\int\limits_{0}^{\lambda_{_S}}\!e^{\!-\!x\!-\!\frac{\vartheta}{\phi x}}\mathrm{d}x\nonumber\\
&\!\!=\!\sum_{i\!=\!0}^{\infty}\!\frac{\!(\!-\!1)^i\lambda_{_S}^{i\!+\!1}}{i!}\!\!\Bigg[\!\bigg(\!\frac{1}{2}\!+\!\frac{\psi(\theta\!-\!\varrho)}{\sqrt{2\pi}}\!\bigg)\!\mathrm{Ei}\!\big(i\!+\!2,\frac{\varrho}{\lambda_{_S}\!\phi}\!\big)\!-\!\bigg(\!\frac{1}{2}\!+\!\frac{\psi\!(\theta\!-\!\vartheta)}{\sqrt{2\pi}}\!\bigg)\!\mathrm{Ei}\!\big(i\!+\!2,\frac{\vartheta}{\lambda_{_S}\!\phi}\!\big)\!+\!\frac{\psi\phi\!\lambda_{_S}}{\sqrt{2\pi}}\!\bigg(\!\!\mathrm{Ei}\!\big(i\!+\!3,\frac{\vartheta}{\lambda_{_S}\!\phi}\!\big)\!-\!\mathrm{Ei}\!\big(i\!+\!3,\frac{\varrho}{\lambda_{_S}\!\phi}\!\big)\!\bigg)\!\Bigg]\!\label{I3}\\
I_4&=\int\limits_{\lambda_{_S}}^{\infty}\int\limits_{\frac{\varrho}{\lambda_{_S}\phi}}^{\frac{\vartheta}{\lambda_{_S}\phi}}\bigg(\frac{1}{2}-\frac{\psi}{\sqrt{2\pi}}(\phi \lambda_{_S} y-\theta)\bigg)e^{-x-y}\mathrm{d}y\mathrm{d}x=\!\bigg(\!\frac{1}{2}\!+\!\frac{\psi\theta}{\sqrt{2\pi}}\!\bigg)\!\int\limits_{\lambda_{_S}}^{\infty}\!\int\limits_{\frac{\varrho}{\lambda_{_S}\phi}}^{\frac{\vartheta}{\lambda_{_S}\phi}}\!e^{\!-\!x\!-\!y}\mathrm{d}y\mathrm{d}x\!-\!\frac{\psi\phi\lambda_{_S}}{\sqrt{2\pi}}\!\int\limits_{\lambda_{_S}}^{\infty}\!\int\limits_{\frac{\varrho}{\lambda_{_S}\phi}}^{\frac{\vartheta}{\lambda_{_S}\phi}}\!ye^{-x-y}\mathrm{d}y\mathrm{d}x\nonumber\\
&=\bigg(\frac{1}{2}+\frac{\psi\theta}{\sqrt{2\pi}}\bigg)(e^{-\frac{\varrho}{\lambda_{_S}\phi}}-e^{-\frac{\vartheta}{\lambda_{_S}\phi}})e^{-\lambda_{_S}}+\frac{\psi\phi\lambda_{_S}}{\sqrt{2\pi}}\bigg(\big(\frac{\vartheta}{\lambda_{_S}\phi}+1\big)e^{-\frac{\vartheta}{\lambda_{_S}\phi}}-\big(\frac{\varrho}{\lambda_{_S}\phi}+1\big)e^{-\frac{\varrho}{\lambda_{_S}\phi}}\bigg)e^{-\lambda_{_S}}\nonumber\\
&=\bigg(\frac{1}{2}+\frac{\psi}{\sqrt{2\pi}}(\theta-\varrho-\lambda_{_S}\phi)\bigg)e^{-\lambda_{_S}-\frac{\varrho}{\lambda_{_S}\phi}}-\bigg(\frac{1}{2}+\frac{\psi}{\sqrt{2\pi}}(\theta-\vartheta-\lambda_{_S}\phi)\bigg)e^{-\lambda_{_S}-\frac{\vartheta}{\lambda_{_S}\phi}}\label{I4}
\end{align}
\hrule
\end{figure*}
Notice that the representations in series in \eqref{I1} and \eqref{I3} are convergent since all the results come from the Taylor series expansion of $e^{-f(x)}$. Now, substituting \eqref{I1}, \eqref{I2}, \eqref{I3}, \eqref{I4} into \eqref{AP_A12}, and by selecting a finite number of terms in the summation, $i_M+1 $, we attain the approximation in \eqref{pout_fin}. \hfill \qedsymbol
\subsection{Derivation of \eqref{pout_inf}}
Let $z=x\cdot y$, which according to \cite[Eq. 5]{Lopez2.2017} with $m=1$, has PDF given by
\begin{align} 
\label{fz}
f_Z(z) &= 2 \operatorname{K}_0(2\sqrt{z}),\ z>0.
\end{align}
Now, working on \eqref{pout}, which is equal to \eqref{pout_inf} since $p_s=\varphi=0$ from the ideal system assumption, we get
\begin{align}\label{AP0}
p_{_{\mathrm{out}}}&\approx\int\limits_0^{\infty}Q\bigg(\frac{C(\phi z)-k/n}{V(\phi z)/n}\bigg)f_Z(z)\mathrm{d}z 
\nonumber\\
&=2\int\limits_0^{\infty}Q\bigg(\frac{C(\phi z)-k/n}{V(\phi z)/n}\bigg)\operatorname{K}_0(2\sqrt{z})\mathrm{d}z.
\end{align}
Using $Q(f(\phi z))\!\approx\!\Omega(\phi z)\!$ given in \eqref{AP}, \eqref{AP0} approximates to 
\begin{align}
&p_{_{\mathrm{out}}}\approx 2\int_{0}^{\infty}\operatorname{K}_0(2\sqrt{z})\Omega(\phi z)\mathrm{d}z \approx 2\int\limits_{0}^{\varrho/\phi}\operatorname{K}_0(2\sqrt{z})\mathrm{d}z\nonumber\\
&\!\approx\! 2I_5(z)\bigg|_{z\!=\!0}^{z\!=\!\tfrac{\varrho}{\phi}}\!\!\!+\!\biggl(1\!+\!\frac{2\psi\theta}{\sqrt{2\pi}}\biggl)I_5(z)\bigg|_{z\!=\!\tfrac{\varrho}{\phi}}^{z\!=\!\tfrac{\vartheta}{\phi}}\!\!-\!\frac{2\psi \phi}{\sqrt{2\pi}}I_6(z)\bigg|_{z\!=\!\tfrac{\varrho}{\phi}}^{z\!=\!\tfrac{\vartheta}{\phi}},\label{AP1}
\end{align}
where $I_5(z)=\int \operatorname{K}_0(2\sqrt{z})\mathrm{d}z$ and $I_6(z)=\int z\operatorname{K}_0(2\sqrt{z})\mathrm{d}z$ are solved next
\begin{align}\label{Bessel1}
I_5(z)&=\int\! \operatorname{K}_0(2\sqrt{z})\mathrm{d}z\stackrel{(a)}{=}\frac{1}{2}\int\!
q\operatorname{K}_0(q)\mathrm{d}q\stackrel{(b)}{=}-\frac{1}{2}q\operatorname{K}_1(q)+\varsigma\nonumber\\
&\stackrel{(c)}{=}-\sqrt{z}\operatorname{K}_1(2\sqrt{z})+\varsigma,
\end{align}
\begin{align}\label{Bessel2}
&I_6(z)\!\!=\!\!\!\int\!\!\! z\operatorname{K}_0(2\sqrt{z})\mathrm{d}z\!\!\!\stackrel{(d)}{=}\!\!\!\int\!\!\frac{q^3}{8}\!\operatorname{K}_0(q)\mathrm{d}q\!\stackrel{(e)}{=}\!\!\! \int\!\!\!\left(\!\frac{q^3}{8}\!\operatorname{K}_2(q)\!-\!\frac{q^2}{4}\!\operatorname{K}_1(q)\!\!\right)\!\!\mathrm{d}q \nonumber \\ 
&\quad\quad\!\stackrel{(f)}{=}-\frac{q^3}{8}\operatorname{K}_3\!(q)+\frac{1}{4}q^2\operatorname{K}_2(q)+\varsigma\nonumber\\
&\!\stackrel{(g)}{=}\!\!-\!\frac{q^3}{8}
\left(\!\frac{8}{q^2}\!\operatorname{K}_1\!(q)\!+\!\frac{4}{q}\!\operatorname{K}_0\!(q)\!+\!\operatorname{K}_1\!(q)\!\right)\! \!+\!\frac{q^2}{4}\!\left(\!\frac{2}{q}\!\operatorname{K}_1\!(q)\!+\!\operatorname{K}_0\!(q)\!\right)\!+\!\varsigma \nonumber\\
&\quad\quad\!\stackrel{(h)}{=}-\frac{q}{2}\bigg(\frac{q^2}{4}+1\bigg)\operatorname{K}_1(q)-\frac{q^2}{4}\operatorname{K}_0(q)+\varsigma\nonumber\\
&\quad\quad\stackrel{(i)}{=} \sqrt{z}(z+1)\operatorname{K}_1(2\sqrt{z}) \!-\! z\operatorname{K}_0(2\sqrt{z})\!+\!\varsigma,
\end{align}
where $(a)$ comes from the transformation $q=2\sqrt{z}$, $(b)$ is based on $\int q^{t+1}\operatorname{K}_t(q)\mathrm{d}q=-q^{t+1}\operatorname{K}_{t+1}(q)$ \cite{Kreh.2012}, in $(c)$ and $(d)$ we reapply the same property as in $(a)$, $(e)$ 
is reached by using  $\operatorname{K}_{t+1}(q)-\operatorname{K}_{t-1}(q)=\frac{2t}{q}\operatorname{K}_t(q)$\cite{Kreh.2012}, in $(f)$ and $(g)$ we reapply the same properties as in $(b)$ and $(e)$ respectively, obtaining $(h)$ is straightforward through algebraic transformations, $(i)$ comes from using property in $(a)$ again, and $\varsigma$ is a constant.

Substituting \eqref{Bessel1} and \eqref{Bessel2} into \eqref{AP1} yields 
\begin{align}\label{eq:outage_deriv2}
\footnotesize
p_{_{\mathrm{out}}}&\approx- 2\sqrt{z}\operatorname{K}_1(2\sqrt{z})\bigl|_0^{\tfrac{\varrho}{\phi}}-
\left(1+\tfrac{2\psi\theta}{\sqrt{2\pi}}\right) \sqrt{z}\operatorname{K}_1(2\sqrt{z})\bigl|_{\tfrac{\varrho}{\phi}}^{\tfrac{\vartheta}{\phi}}+\nonumber\\
&\quad\quad 
+\frac{2\psi \phi}{\sqrt{2\pi}}\left[\sqrt{z}(z+1)\operatorname{K}_1(2\sqrt{z})+z\operatorname{K}_0 (2\sqrt{z})\right]\Bigl|_{\tfrac{\varrho}{\phi}}^{\tfrac{\vartheta}{\phi}} \nonumber \\
&\!\approx\!1\!-\!\sqrt{\tfrac{\varrho}{\phi}}\bigg[1\!+\!\Big(\tfrac{2\psi}{\sqrt{2\pi}}\Big)\big(\varrho\!+\!\phi\!-\!\theta\big)\bigg]K_1\Big(2\sqrt{\tfrac{\varrho}{\phi}}\Big)+\nonumber\\
&\ \ \  
-\!\tfrac{2\psi\varrho}{\sqrt{2\pi}}K_0\Big(2\sqrt{\tfrac{\varrho}{\phi}}\Big)\!-\!\sqrt{\tfrac{\vartheta}{\phi}}\bigg[1\!-\!\Big(\tfrac{2\psi}{\sqrt{2\pi}}\Big)\big(\vartheta\!+\!\phi\!-\!\theta\big)\bigg]
\cdot\nonumber\\
&\quad\quad \cdot 
K_1\Big(2\sqrt{\tfrac{\vartheta}{\phi}}\Big)\!+\!\tfrac{2\psi\vartheta}{\sqrt{2\pi}}K_0\Big(2\sqrt{\tfrac{\vartheta}{\phi}}\Big),
\end{align}
where the last equality comes from $2 \lim\limits_{z\rightarrow 0}\sqrt{z}\operatorname{K}_1(2\sqrt{z})=2\frac{1}{2}=1$ (easily checked by doing the series expansion at $z=0$). \hfill \qedsymbol

\bibliographystyle{IEEEtran}
\bibliography{IEEEabrv,references}

\end{document}